\documentclass[ba]{imsart}
\pubyear{2024}
\volume{TBA}
\issue{TBA}
\firstpage{1}
\lastpage{1}

\usepackage{amsthm}
\usepackage{amsmath,amssymb}
\usepackage{natbib}
\usepackage[colorlinks,citecolor=blue,urlcolor=blue,filecolor=blue,backref=page]{hyperref}
\usepackage{graphicx}

\startlocaldefs

\usepackage{bm}
\usepackage{color}
\usepackage{rotating}

\usepackage{enumitem}

\usepackage{url} 
\usepackage{upgreek}

\usepackage{dsfont}
\usepackage{subcaption}
\usepackage{caption}

\usepackage{float}
 \usepackage{cancel}

\usepackage{wrapfig}

\renewcommand{\Pi}{\Uppi}

\newtheorem{theorem}{Theorem}[section]
\newtheorem*{theorem-sketch}{Sketch of the concentration result}
\newtheorem{proposition}[theorem]{Proposition}

\newtheorem{remark}{Remark}

\newcommand{\R}{\mathbb{R}}
\newcommand{\N}{\mathbb{N}}
\newcommand{\Z}{\mathbb{Z}}

\def\P{\mathds{P}}
\def\E{\mathds{E}}
\def\var{\mathrm{Var}}
\def\J{\mathcal{J}}

\newcommand{\KL}{\operatorname{KL}} 

\usepackage{xspace}
\usepackage[colorinlistoftodos, color=blue!30!white, 
					,disable
]{todonotes}                                        

\newcommand{\omar}[2][]{\todo[size=\scriptsize,color=red!15!yellow,#1]{\textbf{O:} #2}\xspace}
\newcommand{\sami}[2][]{\todo[size=\scriptsize,color=orange!35!white,#1]{\textbf{S:} #2}\xspace}
\newcommand{\paul}[2][]{\todo[size=\scriptsize,color=blue!25!white,#1]{\textbf{P:} #2}\xspace}
\newcommand{\diego}[2][]{\todo[size=\scriptsize,color=red!30!white,#1]{\textbf{D:} #2}\xspace}

\endlocaldefs

\begin{document}

\def\spacingset#1{\renewcommand{\baselinestretch}%
{#1}\small\normalsize} \spacingset{1}


\begin{frontmatter}


\title{  Meta-analysis of Bayesian analyses}
\runtitle{Meta-analysis of Bayesian analyses}

\begin{aug}

\author{\fnms{Paul} \snm{Blomstedt}\thanksref{addr0,addr1,t0}\ead[label=e1]{first@somewhere.com}},
\author{\fnms{Diego} \snm{Mesquita}\thanksref{addr0,addr2,t0,t1}\ead[label=e2]{second@somewhere.com}},
\author{\fnms{Omar} \snm{Rivasplata}\thanksref{addr3,t0}}, \\
\author{\fnms{Jarno} \snm{Lintusaari}\thanksref{addr0}}, 
\author{\fnms{Tuomas} \snm{Sivula}\thanksref{addr0}},
\author{\fnms{Jukka} \snm{Corander}\thanksref{addr4,addr5}},
\and
\author{\fnms{Samuel} \snm{Kaski}\thanksref{addr0,addr3}}

\runauthor{Blomstedt et al.}

\address[]{arXiv:\href{https://arxiv.org/abs/1904.04484}{1904.04484}
}

\address[addr0]{Department of Computer Science, Aalto University
}

\address[addr1]{WithSecure
}

\address[addr2]{School of Applied Mathematics, Getulio Vargas Foundation
}

\address[addr3]{Department of Computer Science, University of Manchester
}

\address[addr4]{Department of Biostatistics, University of Oslo
}

\address[addr5]{Department of Mathematics and Statistics, University of Helsinki
}


\thankstext{t0}{Equal contribution}
\thankstext{t1}{Corresponding author: Diego Mesquita \texttt{diego.mesquita@fgv.br}}

\end{aug}

\begin{abstract}
Meta-analysis aims to generalize results from multiple related statistical analyses through a combined analysis. While the natural outcome of a Bayesian study is a posterior distribution, traditional Bayesian meta-analyses proceed by combining summary statistics (i.e., point-valued estimates) computed from data. 
In this paper, we develop a framework for combining posterior distributions from multiple related Bayesian studies into a meta-analysis. Importantly, the method is capable of reusing pre-computed posteriors from computationally costly analyses, without needing the implementation details from each study.
Besides providing a consensus across studies, the method enables updating the local posteriors \emph{post-hoc} and therefore refining them by sharing statistical strength between the studies, without rerunning the original analyses. 
We illustrate the wide applicability of the framework by combining results from likelihood-free Bayesian analyses, which would be difficult to carry out using standard methodology.
\paul{test}
\diego{test}
\omar{test}
\sami{test}
\end{abstract}



\end{frontmatter}



\section{Introduction}\label{sec:intro}

Meta-analysis comprises a collection of  methods designed to combine results from related statistical analyses carried out in separate studies. 
Traditionally, the per~study results, taken as inputs for meta-analysis, are summary statistics computed from data, such as point estimates for a treatment's effect size.
To combine point estimates, there exist well-established Bayesian and classical
methodologies \citep[see, e.g.,][and literature referenced therein]{Higgins+others:2009}.
However, while the natural outcome of a Bayesian analysis is a posterior distribution, 
the task of combining posteriors obtained from multiple related studies into a meta-analysis has not been explored in the literature.

In standard Bayesian \emph{random effects meta-analysis}, a summary statistic $D_j$, designed to convey information about a local effect of interest $\theta_j$, is observed for each of a number of studies $j=1,\dots, J$. Furthermore, the studies are assumed to be related via exchangeability of the local effects \citep[see][]{Gelman+others:2013}, making them conditionally independent given a shared overall effect $\varphi$. Note that the presence of $\varphi$ is not a property of the original studies but rather a \emph{post-hoc} assumption made by the meta-analyst.
Given the above assumptions, meta-analysis takes the shape of a hierarchical model:
\begin{alignat*}{2}\label{eq:rema_hm}
 \varphi &\sim Q \nonumber \\
 \theta_j &\sim P_{\varphi} \\
 D_j &\sim F_{\theta_j}. 
\end{alignat*}
In the case that $D_j$ is a direct estimate of $\theta_j$, the distribution $F_{\theta_j}$ is usually modeled as $\mathcal{N}(\theta_j,\hat{\sigma}_j^2)$, with $\hat{\sigma}_j^2$ estimated from data.
Let us denote the density functions of $F_{\theta_j}$, $P_{\varphi}$ and $Q$ as $f_j(\cdot|\theta_j)$, $p(\cdot|\varphi)$ and $q(\cdot)$, respectively.
One of the primary goals of this model is to estimate the overall effect $\varphi$, which has marginal posterior density 
\begin{equation}\label{eq:overall_effect}
q(\varphi|D_1,\ldots,D_J) \propto  \prod_{j=1}^J \left[\int f_j(D_j|\theta_j)p(\theta_j|\varphi) d\theta_j\right] q(\varphi).
\end{equation}

It is worth mentioning that in some cases a meta-analysis can be carried out using
data on the level of individual data points instead of aggregated statistics \citep{riley2010meta}. In this case 
$D_j$ can be taken to represent a set of data points.
However, regardless of the granularity of the data and the ensuing distributional assumptions,
the standard Bayesian meta-analysis still takes the form of a hierarchical model as described above, where 
the inputs for the model are point-valued observations.

In this work, we consider a setting in which the meta-analysis aims to combine \emph{posterior distributions} on the local effects $\theta_j$ instead of summary statistics. 
In other words, instead of being point-valued, the observable result from each study is a distribution. 
While summarizing studies via posterior distributions is most intuitive from a Bayesian perspective, traditional Bayesian meta-analysis using the standard machinery of hierarchical models is not equipped to handle distributional observations.

\subsection{Problem description and proposed solution}
\label{sec:intro_problem_solution}

This paper sets out to create a novel methodology and underlying theory to combine posterior distributions on local effects into a meta-analysis.
To the best of our knowledge, there is no previous work in the literature dealing with the problem just described, and therefore our  solution (outlined next) is unprecedented.

Given this setting, and continuing to assume exchangeability of the local effects,
we wish to address two questions:
(i) How should we utilize locally computed posteriors to update our belief about the overall effect in analogy with Equation~\eqref{eq:overall_effect}? And (ii) does the combined model allow us to refine the local posteriors \emph{post-hoc} through sharing of statistical strength?

Let $\Pi_j$ be a posterior distribution resulting from the $j$th study. We may identify $\Pi_j$ with its density $\pi_j$ when the latter exists, which arguably is the case in many practical situations.
As in conventional random-effects meta-analysis, we assume that the local effects are related via exchangeability,
and therefore the effects of the different studies are conditionally independent given an overall effect $\varphi$. 
Unlike in the conventional random-effects setting, which uses marginalized likelihood factors to update the prior, we propose using expected likelihood factors as functionals of the observed distributions, each factor thus being the integral of  $p(\theta_j|\varphi)$ against the density $\pi_j$. 
Accordingly, the proposed rule to update the prior $q(\varphi)$ on the overall effect $\varphi$, given  $\pi_1,\ldots,\pi_J$, is as follows: 
\begin{equation}\label{eq:modified_overall_effect}
q^*(\varphi|\pi_1,\ldots,\pi_J) \propto \prod_{j=1}^J \left[\int p(\theta_j|\varphi) \pi_j(\theta_j)d\theta_j\right] q(\varphi).
\end{equation}
While this expression may look surprising at first sight, we prefer to suspend judgement at this point and resume its discussion in Section~\ref{sec:theory} where its formal definition is presented and its properties are discussed; let us only highlight at this point that Equation~\eqref{eq:modified_overall_effect} is our answer to question (i) formulated in the beginning of this section.

We refer to the distribution defined by Equation~\eqref{eq:modified_overall_effect} as a \emph{superposterior distribution}  
to convey the notion of it being the result of an update on top of posterior distributions $\pi_1,\ldots,\pi_J$.
Notice that there is no explicit data level below the level of the local effects $\theta_j$ in the update defining $q^*(\cdot|\pi_1,\ldots,\pi_J)$. This is a substantial difference with the marginal posterior of the hierarchical model presented in Equation~\eqref{eq:overall_effect}.
We show in Section~\ref{sec:theory} below that the proposed update emerges from an extended \emph{non-hierarchical} model with measure-valued observations $\Pi_1,\ldots,\Pi_J$.
Interestingly, it can easily be shown that a standard non-hierarchical model (i.e., one with point-valued observations)  can be recovered as a special case by choosing $\Pi_j$ to be a Dirac delta measure; this is discussed further in Section~\ref{sec:theory}, where we also show that Equation~\eqref{eq:modified_overall_effect} retains some basic properties of standard Bayesian posterior updates, such as order-invariance in successive updates and concentration as $J\rightarrow \infty$.

We shall also see later on that the proposed update rule defines a joint distribution over $\varphi$ and $\theta_1,\ldots,\theta_J$, from which we can derive a joint density for $\theta_1,\ldots,\theta_J$ and, in turn, the latter can then be used to update the local posteriors $\pi_j$ by marginalizing over all quantities except the one to be updated (further on this in Section~\ref{sec:sharing_strength}); thus answering question (ii) formulated in the beginning of the section. 
This is an important additional feature of our method, which enables updating the local posteriors to share statistical strength between the studies, without the need of rerunning the original analyses.

\subsection{An illustrative example}
\label{sec:intro_illustrative_examples}

To gain an intuitive understanding of the proposed update rule (i.e., the rule outlined in Equation~\eqref{eq:modified_overall_effect}, see also Equation~\eqref{eq:posterior_uncertain_obs} below), consider the following model:
\begin{align*}
\varphi &\sim \mathrm{Beta}(\alpha,\beta)\\
\theta_j &\sim \mathrm{Bernoulli}(\varphi),\quad j=1\ldots,J.
\end{align*}
Suppose initially (but unrealistically) that the $\theta_j$ are observable. Then, a standard application 
of Bayes' rule to update the prior distribution $\mathrm{Beta}(\alpha,\beta)$, conditional on observed values $\theta_j$, yields a posterior distribution which is 
$\mathrm{Beta}(\alpha+r,\beta+J-r)$ with $r=\sum_{j=1}^J \theta_j$, by conjugacy. 
For reference, if $q(\varphi)$ denotes the density of the $\mathrm{Beta}(\alpha,\beta)$ prior, we note that
the posterior density for this model takes the form
\[
q(\varphi|\theta_1,\ldots,\theta_J) \propto \prod_{j=1}^J\varphi^{\theta_j} (1-\varphi)^{1-\theta_j}q(\varphi).
\] 
Next, suppose that the values of $\theta_j$ cannot be observed, but instead, the observables are distributions on these variables. 
Distributions on unobservable quantities are in practice often taken to be continuous, 
but for the purpose of this example we let the distributions on $\theta_j$ be discrete. Specifically, we assume they are binary distributions.

Thus, let the binary distributions $\pi_j$ be the observables, where $\pi_j(k)=\Pi_j(\theta_j=k)$, $k\in \{0,1\}$, and $\sum_{k=0}^1 \Pi_j(\theta_j=k)=1$.
We now wish to update the prior $\mathrm{Beta}(\alpha,\beta)$ conditional on the observed $\pi_j$.  
Following Equation~\eqref{eq:modified_overall_effect},
the update takes the form 
\[
q^*(\varphi|\pi_1,\ldots,\pi_J) \propto \prod_{j=1}^J \left[\hspace{1pt}\sum_{k=0}^1 \varphi^{k} (1-\varphi)^{(1-k)} \pi_j(k)\right] q(\varphi).
\]
While the modified model no longer permits analytical calculations through conjugacy in general, two analytically tractable special cases can be identified. 
Specifically, if $\pi_j(0)=\pi_j(1)=0.5$ for all $j$, then the observed distributions are uninformative about the value of $\theta_j$, and $q^*(\cdot|\pi_1,\ldots,\pi_J)$ coincides with the prior $\mathrm{Beta}(\alpha,\beta)$. 
On the other hand, if for each $j$ there is a $k_j$ with $\Pi_j(\theta_j=k_j)=1$,
then the $\theta_j$'s are in effect fully observed, and $q^*(\cdot|\pi_1,\ldots,\pi_J)$ from Equation~\eqref{eq:modified_overall_effect} equals the standard posterior $\mathrm{Beta}(\alpha+r,\beta+J-r)$. 
A numerical illustration is provided in Figure~1.

\begin{figure}[H]
\centering
\includegraphics[width=0.65\linewidth]{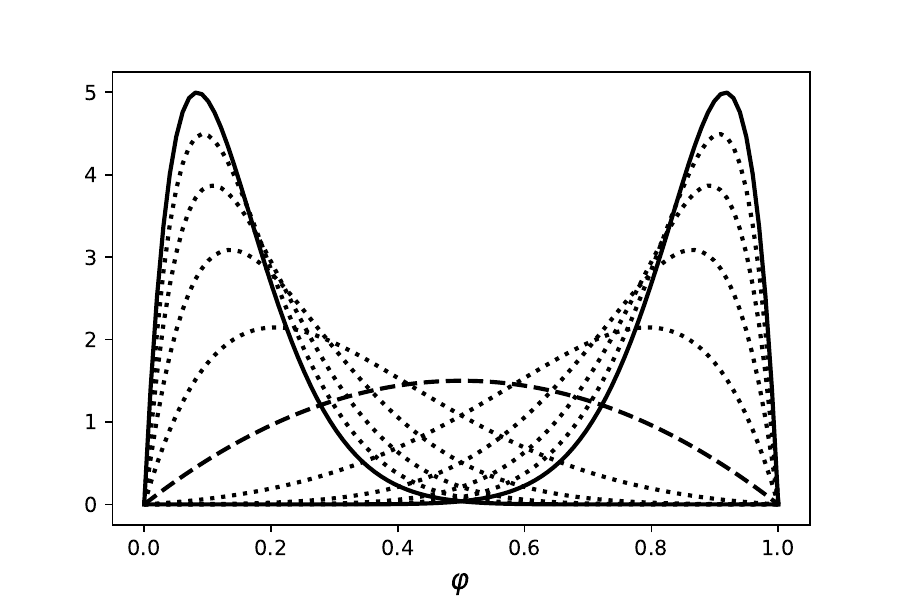}
\caption{
Updated density functions for $\varphi$ with updates as per Equation~\eqref{eq:modified_overall_effect}, each computed under the model $\varphi \sim \mathrm{Beta}(2,2)$, $\theta_j \sim \mathrm{Bernoulli}(\varphi)$, $j=1,\ldots,10$. The distributional observations are $\pi_j(k)=\Pi_j(\theta_j=k)$, $k\in \{0,1\}$, assumed to be identical for all $j$.  
The different densities are obtained by varying $\pi_j(1)$ from 0 to 1 by increments of 0.1.
The solid curves correspond to $\pi_j(1)=0$ and $\pi_j(1)=1$, equivalent to posteriors computed conditional on all $\theta_j$ fully observed with values 0 and 1, respectively. The dashed curve corresponds to $\pi_j(1)=0.5$, and equals the prior $\mathrm{Beta}(2,2)$.}
\label{fig:bernoulli_example_intro}
\end{figure}

\subsection{Further discussion of the setting}
\label{sec:intro_further_discussion}

Our framework builds on the same general idea of updating the distribution of a parameter, conditional on observations, as classical Bayesian inference. 
In a meta-analysis context, the distribution to be updated is the distribution of the overall effect $\varphi$.
However, in our setting the observations themselves are distributions, namely, distributions on the local effects $\theta_j$; whereas in previous meta-analysis settings the observations were summary statistics on the local effects.
Thus, our setting deals with distributional observations, and the observed distributions may be interpreted as expressing beliefs about the values of the local effects of interest.
Apart from being direct representations of the uncertainty inherent in parameter estimates, posterior distributions may include prior knowledge not present in the data, possibly obtained by expert elicitation \citep[e.g.][]{Albert+others:2012}. This is particularly important in problems where not enough data are available to inform a model about some of its parameters \citep[e.g.][]{Kuikka+others:2014}.
This novel setting is therefore naturally suited for a meta-analysis that aims to combine the posterior distributions produced by related Bayesian analyses.

An advantage of processing study-specific information as readily computed posterior distributions, 
instead of re-computing the local posteriors within a hierarchical model, is that the complexities 
regarding inference and computation in each study can be abstracted away from the meta-analysis.
This is an asset since the meta-analyst may neither have access to the study-specific data or the necessary implementation details, nor the computational resources to rerun each analysis as part of a larger hierarchical model that combines them. 
For example, in Section~\ref{subsec:tbc}, we illustrate our method in a case where the local data models lack a tractable likelihood function. While extending likelihood-free inference to hierarchical models is possible, it is computationally costly \citep{turner2014hierarchical},
and consequently, building a hierarchical model from each dataset would be very challenging.
It is important to note that our discussion here focuses specifically on scenarios where computationally complex posterior distributions are readily available, rather than serving as a comprehensive comparison between our framework and hierarchical models in general.

\subsection{Summary of contributions}

The main contributions of our work are the following:
\begin{itemize}[itemsep=2pt,topsep=1pt]
    \item Developing a framework for meta-analysis that combines posterior distributions. The framework proposes a rule to combine locally computed posteriors to update beliefs on an overall effect. The framework also shows how to refine the local posteriors \emph{post-hoc} through sharing of statistical strength.
    \item Demonstrating the framework via numerical examples with both synthetic and real-world data. To illustrate the generality of the approach, we consider in the numerical examples the problem of combining results from analyses conducted using likelihood-free Bayesian analyses, where there is no closed-form likelihood to relate data or summary statistics to the quantity or effect of interest, which poses a challenge for traditional formulations of meta-analysis.
    \item Proving formally the properties of the proposed rule. This entails the development of a theoretical framework to justify the proposed rule, and proving that it satisfies the properties of posterior update rules such as order-invariance in successive updates and concentration in the limit of infinite observations.
\end{itemize}

\subsection{Organization of the paper}

The paper is structured as follows:  
Section~\ref{sec:mba} summarizes the main equations for practical use and provides further insight into the framework. 
Importantly, we also discuss the use of our framework to refine study-specific posteriors \emph{post-hoc} through sharing of statistical strength across studies, and we draw connections to standard Bayesian meta-analysis.
Section~\ref{sec:numerical_demos} illustrates our method with posterior distributions resulting from likelihood-free Bayesian analyses, which are typically computationally complex and may require careful hyperparameter tuning.
Section~\ref{sec:theory} develops a novel framework for Bayesian inference with distributional observations, which is the basis for our novel meta-analysis approach. Section~\ref{sec:related_work} presents a discussion of related work.
The paper ends with concluding remarks in Section~\ref{sec:discussion}. 
The appendices (in the supplementary materials) provide proofs, an interpretation of our approach as message-passing in probabilistic graphical models, a computational strategy for our method, which can be implemented using general-purpose software, and a further experiment on a synthetic dataset.

\section{Meta-analysis of Bayesian analyses}\label{sec:mba}

We are motivated by the problem of conducting meta-analysis for Bayesian analyses summarized as posterior distributions, and refer to our framework as \emph{meta-analysis of Bayesian analyses} (MBA). 
The central updates of the framework are given in Equations (\ref{eq:posterior_uncertain_obs_density}) and (\ref{eq:update_pi_j}) below, which update beliefs on global and local effects, respectively.

To reiterate, a meta-analysis aims to combine inferences produced by various related studies into a consensus analysis. 
Thus, a local effect of interest $\theta_j$ is targeted by each of a number of studies $j=1,\dots,J$.
Judging the quantities $\theta_j$ to be exchangeable, the \emph{meta-analyst} may formulate a model 
\begin{equation}\label{eq:meta-analyst}
\prod_{j=1}^J p(\theta_j|\varphi)q(\varphi), 
\end{equation}
with an appropriate prior $q$ on the parameter $\varphi$. 
Note that this model is formulated \emph{as if} the $\theta_j$'s were fully observable quantities; this has an important role for conceptualization but is unrealistic in practical situations.

\subsection{Updating belief on the overall effect}
In our setting of interest, the observables consist of a set of posterior density functions $\{\pi_1,\ldots,\pi_J\}$, each informing us about beliefs on the value of a corresponding quantity of interest $\{\theta_1,\ldots,\theta_J\}$.
Then, to update $q$ based on the observed posterior densities, we apply the rule proposed in Equation~\eqref{eq:modified_overall_effect}, recalled here for convenience:
\begin{equation}\label{eq:posterior_uncertain_obs_density}
q^*(\varphi) \propto \prod_{j=1}^J \left[\int_{\Theta} p(\theta_j|\varphi)\pi_j(\theta_j) d\theta_j\right] q(\varphi),
\end{equation}
where we used the notation $q^*(\cdot):=q^*(\cdot|\pi_1,\ldots,\pi_J)$ for brevity. 
The formal definition of the posterior distribution is given in Section~\ref{sec:theory} (cf. Equation~\eqref{eq:posterior_uncertain_obs}), from which the density form given in Equation~\eqref{eq:posterior_uncertain_obs_density} is obtained.

\begin{remark}[Fixed-effects MBA] 
We can derive a fixed-effects version of MBA as a special case when the $\Phi = \Theta$ and the likelihood $p(\cdot | \varphi)$ approaches a point mass at $\varphi$. In this case, the MBA posterior for $\varphi$ becomes proportional to the product of study-specific distributions $\pi_1,\ldots,\pi_J$ and the prior $q$, i.e., 
\begin{equation*}
    q^{*}(\varphi) \rightarrow Z^{-1} \prod_{j=1}^{J} \pi_j(\varphi) q(\varphi).
\end{equation*}
Note that when the prior $q(\varphi)$ is uniform, $q^{*}(\varphi)$ becomes proportional to the product of study-specific distributions. If we also assume the $\pi_1, \ldots, \pi_J$ are Bayesian posteriors computed using flat priors, then $q^{*}(\varphi)$ is equivalent to multiplying the study-specific likelihoods and normalizing the result.
\end{remark}

\subsection{Post-hoc refinement of local posteriors}\label{sec:sharing_strength}

In a meta-analysis context, the parameter $\varphi$ often has an interpretation as the central tendency of some shared property of $\theta_1,\ldots,\theta_J$, such as the mean or the covariance (or both jointly).  
As such, inference on $\varphi$ is often of primary interest in providing a `consensus' over a number of studies.  
As a secondary goal, we may also be interested in updating the posterior of 
any individual quantity $\theta_j$, subject to the information provided by the posteriors on the remaining quantities. To do so, we may use the joint density of $\theta_1,\ldots,\theta_J$ as starting point:
\begin{align*}
    p^*(\theta_1,\ldots,\theta_J) \propto   \int_{\Phi}\prod_{j=1}^J \left[ p(\theta_j|\varphi)\pi_j(\theta_j) \right] q(\varphi) d\varphi.
\end{align*}
The formal definition of the joint distribution is given in Section~\ref{sec:theory} (cf. Equation~\eqref{eq:joint_belief}), from which the density form given above is obtained.
To update a local posterior, one can then marginalize over all quantities but the one to be updated. Let $\J := \{1,\ldots,J\}$ be a set of indices and let $j'\in \J$ be an arbitrary index in this set.
The density function $\pi_{j'}$ is then updated as follows:   
\begin{equation}\label{eq:update_pi_j}
\pi_{j'}^*(\theta_{j'}) \propto 
\int_{\Phi}p(\theta_{j'}|\varphi)\pi_{j'}(\theta_{j'})\prod_{j\in\J\setminus j'}^J \left[\int_{\Theta} p(\theta_j|\varphi)\pi_j(\theta_j) d\theta_j\right] q(\varphi) d\varphi.
\end{equation}

It turns out that Equations (\ref{eq:posterior_uncertain_obs_density}) and (\ref{eq:update_pi_j}) can also be formulated as an instance of message-passing in probabilistic graphical models \cite[e.g.][]{Yedidia+others:2001}. We discuss this view further in the supplementary material (Appendix C).

\begin{remark}[Non-degeneracy of local effects]
According to Section~\ref{sec:posterior_concentration}, the density $q^*(\varphi)$, defined in Equation~\eqref{eq:posterior_uncertain_obs_density}, will under suitable conditions become increasingly peaked around some point $\varphi_0$, as $J\rightarrow \infty$. 
That $\pi_{j'}^*(\theta_j)$ does not behave similarly, becomes clear by the following considerations.  
First, we note that Equation~\eqref{eq:update_pi_j} is equivalent to
\begin{equation*}
    \pi_{j^\prime}^{*}(\theta_{j^\prime}) = Z_{j^\prime}^{-1}\pi_{j^\prime}(\theta_{j^\prime}) \int_{\Phi} p(\theta_{j^\prime} | \varphi) q^*(\varphi | \pi_1, \ldots, \pi_{j^\prime - 1}, \pi_{j^\prime + 1}, \ldots, \pi_{J}) \,d \varphi,
\end{equation*}
where $Z_j$ is a normalizing constant. 
As $q^*(\varphi | \pi_1, \ldots, \pi_{j^\prime - 1}, \pi_{j^\prime + 1}, \ldots, \pi_{J})$ becomes increasingly peaked around $\varphi_0$, the integral in the above equation converges to $p(\theta_{j^\prime} | \varphi_0)$. Consequently,
\begin{equation*}
    \pi_{j^\prime}^{*}(\theta_{j^\prime}) \rightarrow Z_{j^\prime}^{-1} \pi_{j^\prime}(\theta_{j^\prime}) p(\theta_{j^\prime} | \varphi_0)    ,
\end{equation*}
which can only be degenerate if either $\pi_{j^\prime}(\theta_{j^\prime})$ or  $p(\theta_{j^\prime} | \varphi_0)$ is degenerate by design.
Instead of degeneracy, $\pi_{j^\prime}^{*}(\theta_{j^\prime})$ exhibits \emph{shrinkage} with respect to $\varphi_0$. 
\end{remark}

\subsection{Bayesian meta-analysis as a special case}
\label{sec:rema_fema}

It is straightforward to show that standard Bayesian random-effects and fixed-effects meta-analyses can be recovered as special cases of MBA.  
In its traditional formulation \cite[e.g.][]{Normand:1999}, \emph{random-effects} meta-analysis (REMA) assumes that for each of $J$ studies, a summary statistic, $D_j$, $j=1,\ldots,J$, has been observed, drawn from a distribution with study-specific mean $\E(D_j)=\theta_j$ and variance $\var(D_j)=\sigma_j^2$:
\begin{equation}\label{eq:REMA_data_model}
D_j \sim \mathcal{N}(\theta_j,\sigma_j^2),
\end{equation}
where the approximation of the distribution of $D_j$ by a normal distribution is justified by the asymptotic normality of maximum likelihood estimates. The variances $\sigma_j^2$ are directly estimated from the data, while the means $\theta_j$ are assumed to be drawn from some common distribution, typically 
\[
\theta_j \sim \mathcal{N}(\mu,\sigma_0^2), 
\]
where the parameters $\mu$ and $\sigma_0^2$ represent the average treatment effect and inter-study variation, respectively. 
For this model, the $(\mu, \sigma_0^2)$ are the global effect parameters, for which the general notation $\varphi$ is used elsewhere in the text. 
\emph{Fixed-effects} meta-analysis (FEMA) is a special case of REMA, where $\sigma_0^2 \rightarrow 0$, such that $\theta_1=\theta_2= \cdots=\theta_J$. 

The marginal posterior density for the parameters $(\mu,\sigma_0^2)$ in REMA can be written as 
\begin{alignat*}{2}
q(\mu,\sigma_0^2|D_1,\ldots,D_J) 
&\propto q(\mu,\sigma_0^2)\prod_{j=1}^J \int_{\Theta} N(D_j|\theta_j,\hat{\sigma}_j^2) N(\theta_j|\mu,\sigma_0^2)d\theta_j\\ 
 &\propto q(\mu,\sigma_0^2)\prod_{j=1}^J \int_{\Theta} l_j(\theta_j;D_j) N(\theta_j|\mu,\sigma_0^2) d\theta_j ,
\end{alignat*}
where $N(\cdot|\cdot,\cdot)$ denotes a Gaussian density function, $l_j(\theta_j;D_j)$ is the likelihood function of $\theta_j$ given $D_j$, and $\hat{\sigma}_j^2$ is the empirical variance of $D_j$. 

To study the connection between the above posterior density and  Equation~\eqref{eq:posterior_uncertain_obs_density}, assume that instead of a summary statistic $D_j$, each study has been summarized using a posterior distribution with density $\pi_j(\theta_j)$ over its study-specific effect parameter $\theta_j$. If each of the posteriors has been independently computed under a Gaussian data model (see Equation~\eqref{eq:REMA_data_model}), and using  improper uniform priors 
$\nu_j(\theta_j)\propto 1$ for each $j=1\ldots J$, then the local posterior densities are 
\[
\pi_j(\theta_j)=N(\theta_j|D_j,\hat{\sigma}_j^2)\propto \exp\left\lbrace-\frac{(D_j-\theta_j)^2}{2\hat{\sigma}_j^2}\right\rbrace = l_j(\theta_j;D_j) \quad \forall j=1,\ldots,J,
\] 
resulting in both the marginal posterior of $(\mu,\sigma_0^2)$ for REMA, and the corresponding superposterior for MBA, being equivalent.

\section{Numerical studies}\label{sec:numerical_demos}

In this section, we illustrate the proposed  
novel framework, {meta-analysis of Bayesian analyses} (MBA), with numerical examples\footnote{Code available at \href{https://github.com/weakly-informative/BA-MBA}{github.com/weakly-informative/BA-MBA}}. 
To illustrate the generality of the approach,
we consider the problem of combining results from analyses conducted using likelihood-free models. 
We emphasize, however, that MBA is not restricted to this class of models.
In likelihood-free inference, the data can typically be summarized by a number of different statistics but there is no closed-form likelihood to relate these to the quantity or effect of interest, which poses a challenge for traditional formulations of meta-analysis.  
Using our framework, we directly utilize the local posteriors to build a joint model, from which we sample using a simple computational strategy, described in the supplementary material (Appendix D). 
In addition to modeling the shared central tendency of the inferred model parameters, we demonstrate that weakly informative or poorly identifiable posteriors for original studies can be updated \emph{post-hoc} through joint modeling.

We begin with a brief review of likelihood-free inference using approximate Bayesian computation.

\subsection{Likelihood-free inference using approximate Bayesian computation}\label{sec:abc}

Approximate Bayesian computation (ABC) is a paradigm for Bayesian inference in models which either entirely lack an analytically tractable likelihood function, or for which it is costly to compute.
The only requirement is that we are able to sample data from the model by fixing values for the parameters of interest, as is the case for simulator-based models. 
In the basic form of ABC, simulations are run for parameter proposals drawn from a prior distribution.
The parameter proposals whose simulated data $x_\theta$ match the observed data $x_0$ are collected and constitute a sample from the posterior distribution. 
It can be shown that this process is equivalent to accepting parameter proposals in proportion to their likelihood value, given the observed data,
as is done in traditional rejection sampling.

In practice, the simulated data virtually never exactly match the observed data and very few samples 
from the posterior distribution would be acquired.
This problem can be circumvented by loosening the acceptance condition to accept samples whose simulations yield results similar enough to the observed data.
For this purpose, a dissimilarity function $d$ and a scalar $\varepsilon > 0$ are defined such that a parameter proposal with respective simulation result $x_\theta$ is accepted if $ d(x_\theta, x_0) \leq \varepsilon $. This function is often defined in terms of summary statistics $s(x_\theta)$ and $s(x_0)$. For example, $d$ could be defined as the Euclidean distance between $s(x_\theta)$ and $s(x_0)$.
The aforementioned relaxation results in samples being drawn from
an \emph{approximate} posterior instead of the actual posterior distribution, hence the name approximate Bayesian computation. 
 
For a comprehensive introduction to ABC, see \citet{Marin+others:2012}. More recent developments are reviewed in \cite{Lintusaari+others:2017}.
In the following numerical illustrations, the inputs for meta-analysis are posterior distributions for the separate local studies obtained using ABC.
These likelihood-free inferences are implemented using the ELFI open-source software package \citep{Lintusaari+others:2018}.

\subsection{Example 1: MA\texorpdfstring{$(q)$}{(\textit{q})} process}
\label{subsec:ma2}

The moving-average process of order $q$, abbreviated MA$(q)$, 
is a standard example in the literature on likelihood-free inference due to its simple structure but fairly complex likelihood and non-trivial relationship between parameters and observed data.
In our first example, we use simulated data from a MA$(q)$ process of order $q=2$. 
Assuming zero mean, the process $(y_t)_{t\in \N^{+}}$ is defined as  
\begin{equation}\label{eq:ma2}
y_t = \epsilon_t + \theta^{(1)} \epsilon_{t-1} + \theta^{(2)} \epsilon_{t-2},
\end{equation}
where $(\theta^{(1)},\theta^{(2)}) \in\R^2$ 
and $\epsilon_s \sim \mathcal{N}(0,1)$, ${s\in\Z}$.\footnote{To keep the notation consistent with Sections~\ref{sec:mba} and \ref{sec:theory}, we use subscripts of $\theta$ to index studies, while different dimensions of $\theta$ are indexed using superscipts.} 
The quantity of interest for which we conduct inference is $\theta = (\theta^{(1)},\theta^{(2)})$. 
Following \citet{Marin+others:2012}, we use as prior for $\theta$ a uniform distribution over the set
\[
\mathcal{T}\subset \R^2 \triangleq \{(\theta^{(1)},\theta^{(2)})\in\R^2|-(\theta^{(2)} + 1) < \theta^{(1)} < \theta^{(2)} + 1, \; -1<\theta^{(2)}<1\},
\]
which, by restriction of the parameter space, imposes a general identifiability condition for MA$(q)$  processes. 
Inference for $\theta$ is then conducted using ABC with rejection sampling, taking as summary statistics the empirical autocovariances of lags one and two, denoted as $\hat{\gamma}_1$ and $\hat{\gamma}_2$, respectively. 
Euclidean distance of 0.1 is used as the acceptance threshold. 
Note that in this example, and in contrast to the next example in Section~\ref{subsec:tbc}, 
posterior sampling from the current model is also possible using MCMC, although evaluation of the 
likelihood is difficult due to the unobserved $\epsilon_s$.

To illustrate our meta-analysis framework, we first sample $J=12$ realizations of $\theta$ using the following generating process: 
\begin{equation}\label{eq:theta_gen}
\theta^{(1)}\sim \mathrm{Unif}(0.4,0.8), \quad \theta^{(2)} \sim \mathcal{N}(-0.4+\theta^{(1)}, 0.04^2).
\end{equation}
Given each realization $\theta_j=(\theta^{(1)}_j,\theta^{(2)}_j)$, $j=1,\ldots,J$, we then generate a series of $10$ data points, $(y_{j,1},\ldots,y_{j,10})$, according to Equation~\eqref{eq:ma2}. 
For each time-series, we independently conduct likelihood-free inference as described above, generating $1000$ samples from the posterior.
The computed posterior distributions along with their corresponding true parameter values are shown in Figure~\ref{fig:ma2_independent_posteriors}.  
It can be seen that the very limited information given by the data in each of the analyses leaves the posteriors with a considerable uncertainty.

For meta-analysis, we first specify a model for the study-specific effects $\theta_1,\ldots,\theta_J$ \emph{as if} they were observed quantities from an exchangeable sequence; see Equation~\eqref{eq:meta-analyst}. As the true generating mechanism of the effects is typically unknown, the model must be specified according to the analyst's judgment. To reflect this, we will here model the generating process as a Gaussian distribution with parameters $\varphi = (\mu,\Sigma_0)$, 
\begin{equation}\label{eq:theta_conditional}
\theta_j\sim \mathcal{N}_2(\mu, \Sigma_0).
\end{equation}
For $\mu$ and the covariance matrix $\Sigma_0$, we use Gaussian and inverse Wishart priors, respectively,
\begin{equation}\label{eq:hyperpriors}
\centering
 \mu \sim \mathcal{N}_2(m,V) \quad \text{ and } \quad \Sigma_0\sim \mathcal{W}^{-1}(\nu,\Psi) \,,
\end{equation} 
with 
\[
\centering
m = \begin{bmatrix}
1/2\\
0
\end{bmatrix},\quad 
V = \begin{bmatrix}
0.4 & 0.05 \\
0.05 & 0.1
\end{bmatrix},\quad  \text{ and } \quad
\nu = 4,\quad  
\Psi = \begin{bmatrix}
0.4 & 0.1 \\
0.1 & 0.2
\end{bmatrix}.
\]
The above values were chosen to provide reasonable coverage of $\mathcal{T}$, the constrained support of $\theta$. Furthermore, $\nu$ was chosen as $\mathrm{dim}(\theta) + 2$ to directly yield $\Psi$ as the mean of the inverse Wishart prior on $\Sigma_0$. 

After specifying the assumed generative model for $\theta_1,\ldots,\theta_J$ according to Equations~(\ref{eq:theta_conditional})--(\ref{eq:hyperpriors}), the following step is to incorporate the observed posteriors for each $\theta_j$ into the inference. 
We do this using the computational scheme detailed in Appendix D.
Since the observed posteriors can be of arbitrary form, we carry out the computation in two stages. First, we draw samples from an approximate joint distribution, where each observed posterior is approximated with a suitable parametric distribution. Then, in the second stage, we refine the joint distribution using 
sampling/importance resampling \citep[SIR;][]{Smith+Gelfand:1992} to correct for discrepancies between the 
approximated and true forms of the posteriors. In the current example, we apply the MBA updates using 
bivariate normal approximations to each of the $J=12$ posteriors. 

We compare MBA against results obtained using traditional random-effects meta-analysis (REMA), as specified in Section~\ref{sec:rema_fema}. 
The likelihood of REMA is given by the model
\begin{equation}\label{eq:ma2_rema_likelihood}
\hat{\!\theta}_j \sim  \mathcal{N}_2(\theta_j, \hat{\Sigma}_j),
\end{equation}
where the effect estimates $\,\hat{\!\theta}_j=(\hat{\theta}^{(1)}_j,\hat{\theta}^{(2)}_j)$ are computed numerically using conditional sum of squares\footnote{Implemented in the \texttt{statsmodels} Python module.} and the study-specific covariance matrices $\hat{\Sigma}_j$ are estimated using bootstrap. The hierarchical distribution on $\theta_j$ follows Equations (\ref{eq:theta_conditional}) and (\ref{eq:hyperpriors}) above. Therefore, the essential difference between REMA and MBA is whether we combine this distribution with a likelihood function based on Equation~\eqref{eq:ma2_rema_likelihood} or with the observed posterior on $\theta_j$. 
The results of the comparison are presented in Figures \ref{fig:ma2_mu0_posteriors}--\ref{fig:ma2_re_posteriors}.\footnote{The experiment was repeated with multiple random seeds, yielding similar results.}

\begin{figure}[H]
\centering
\includegraphics[width=0.83\linewidth]{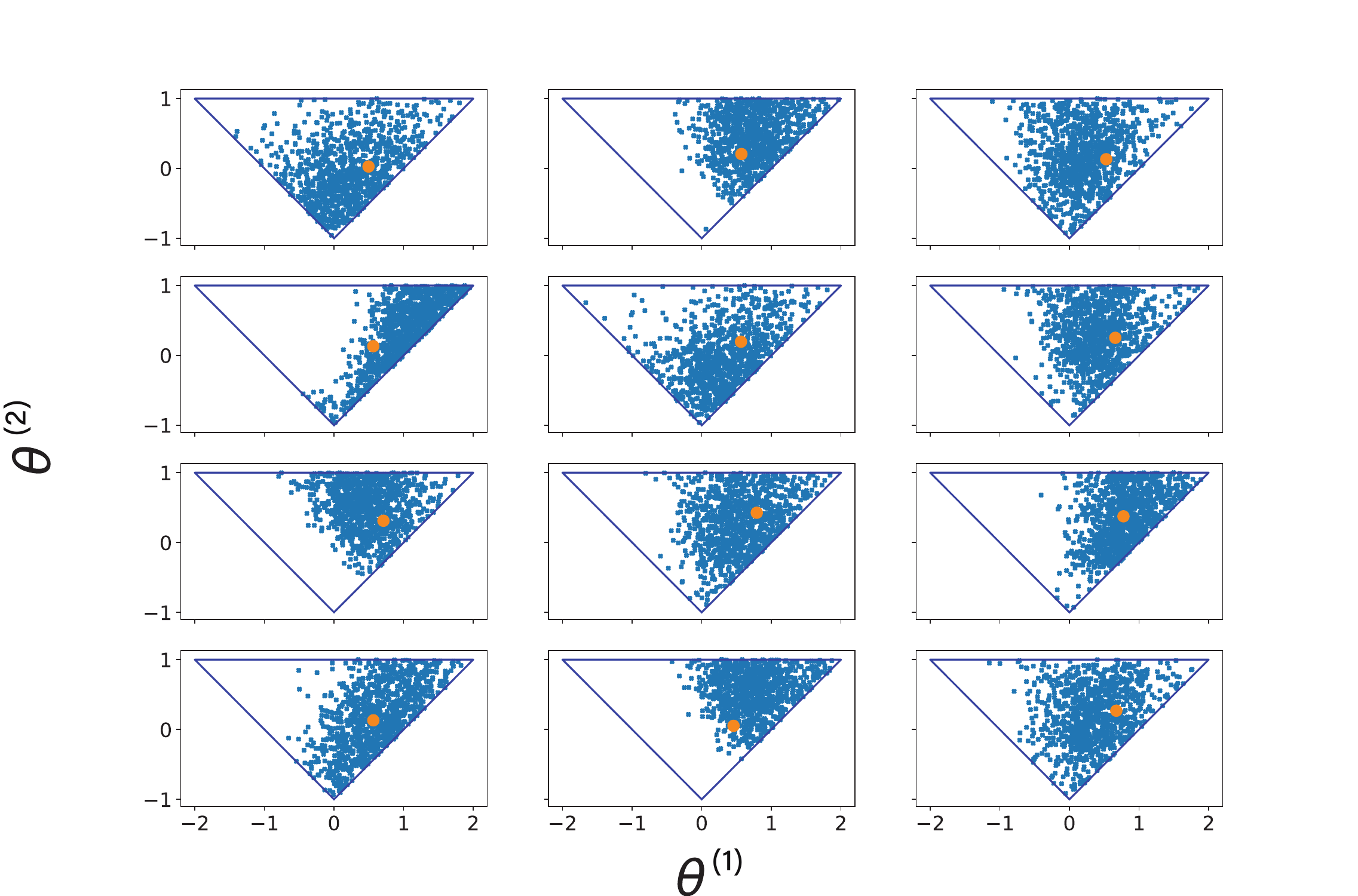}
\\
\caption{Posterior samples for the parameters $(\theta^{(1)}_j,\theta^{(2)}_j)$ of $J=12$ MA$(2)$ processes, obtained using rejection-sampling ABC. The yellow dots denote the true generating parameter values.}
\label{fig:ma2_independent_posteriors}
\end{figure}

\begin{figure}[H]
\centering
\includegraphics[width=0.83\linewidth]{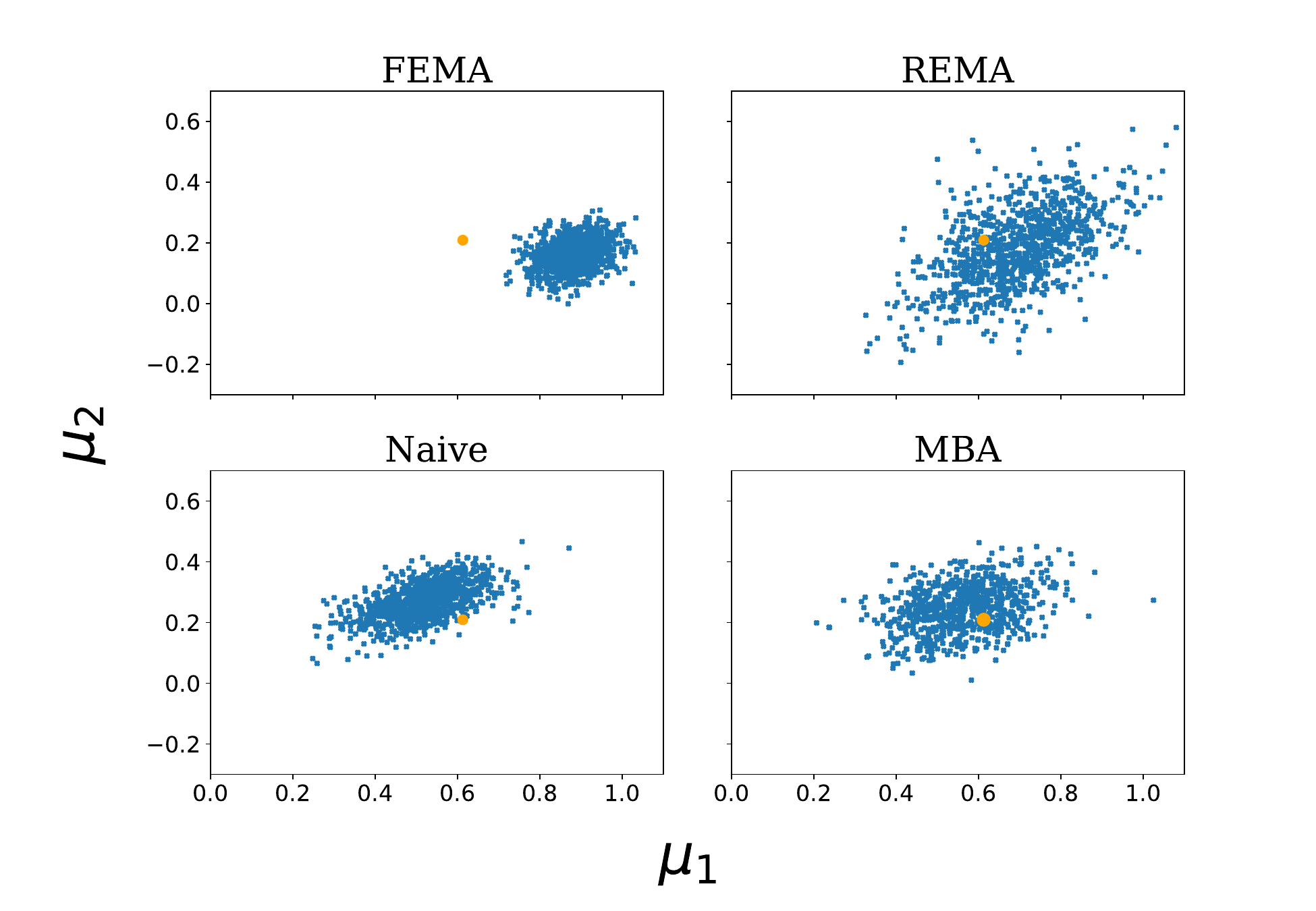}
\\
\caption{Posterior for the global  effect $\mu$ using FEMA, REMA, a `naive' model, and MBA. The yellow dot denotes the actual mean of the distribution used to generate the $\theta_j$ values for the MA(2) series.}
\label{fig:ma2_mu0_posteriors}
\end{figure}

\begin{figure}[H]
\centering
\includegraphics[width=0.83\linewidth]{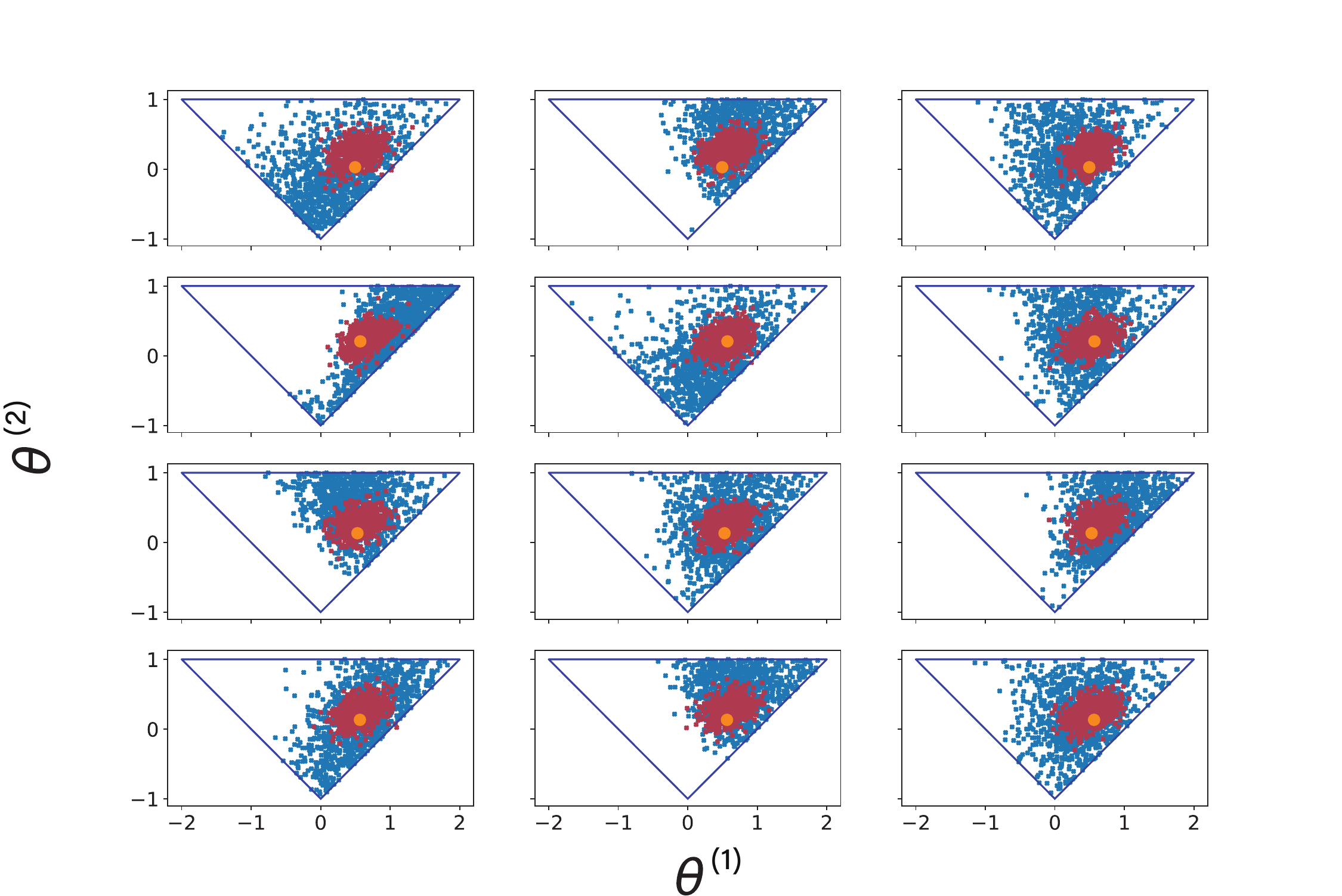}
\\
\caption{Updated posteriors for study-specific effects obtained using MBA (red) on top of the ones originally obtained using independent ABC's (blue). Yellow dots denote the true  parameter values.}
\label{fig:ma2_updated_posteriors}
\end{figure}

\begin{figure}[H]
\centering
\includegraphics[width=0.83\linewidth]{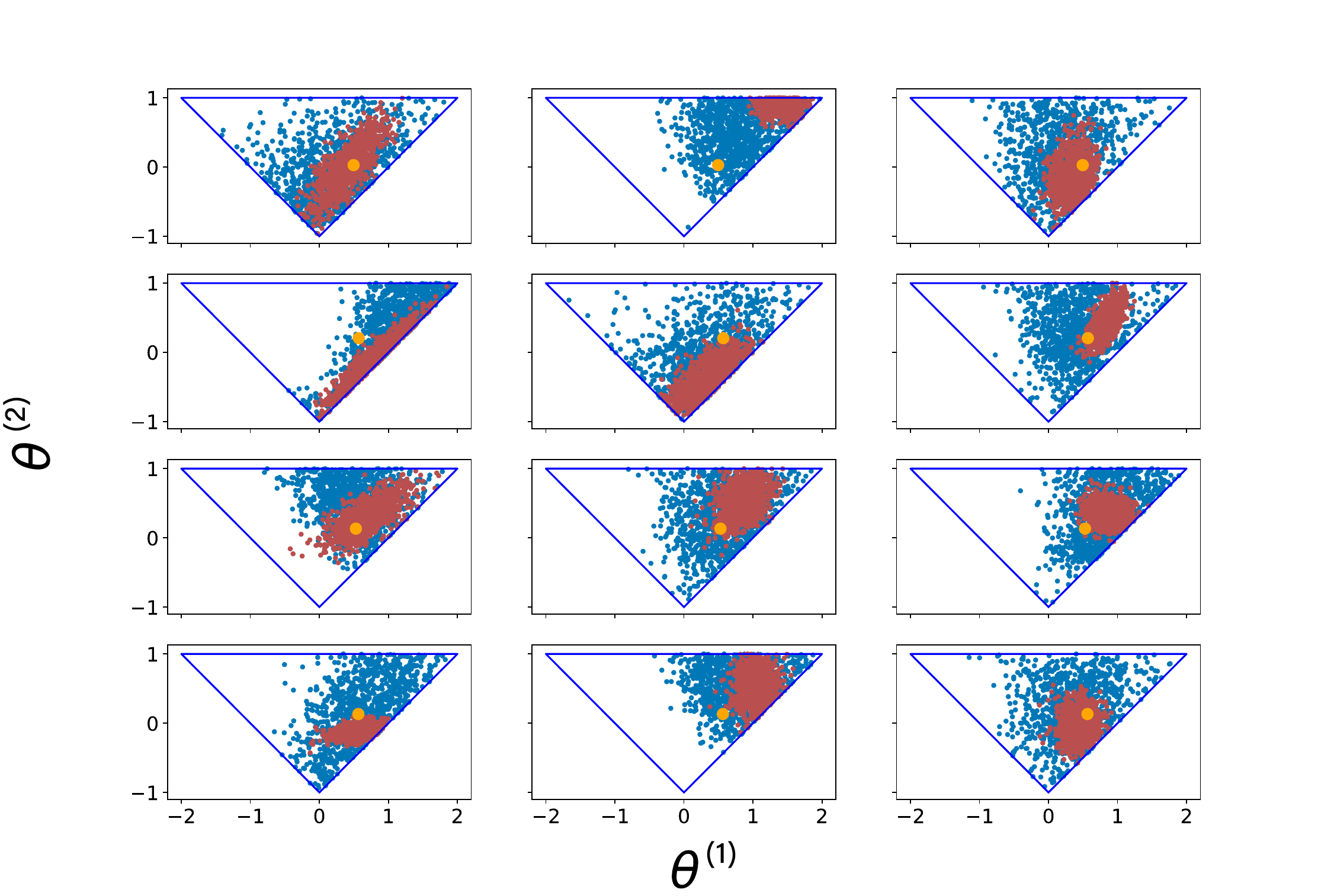}
\\
\caption{Posterior marginals for  study-specific effects obtained using REMA (red) on top of the ones originally obtained using independent ABC's (blue). Yellow dots denote the true  parameter values.}
\label{fig:ma2_re_posteriors}
\end{figure}

Figure~\ref{fig:ma2_mu0_posteriors} shows the posterior distribution (superposterior in the case of MBA) for the global mean effect $\mu$, obtained using four different models: in addition to MBA and REMA, we used fixed-effects meta-analysis (FEMA, which is a special case of REMA, see Section~\ref{sec:rema_fema}), and a `naive' model corresponding to ordinary Bayesian inference using the means of the observed posteriors on $\theta_j$ as observed data.
In Section~\ref{sec:theory}, we formulate the theoretical framework underlying MBA as an extension of Bayes' rule from point-valued observations to distributional observations.
Notice that the naive model can be thought of as a special case of MBA, where the probability mass of each of the observed posteriors concentrate at a single point, namely the posterior mean. 

As expected, FEMA is clearly inappropriate in this situation, and results in a heavily biased and overly confident posterior. The naive model is less biased than FEMA as it makes use of information contained in the observed posteriors, but compared to MBA, it does not properly account for the uncertainty contained in them.
Both REMA and MBA result in posteriors with more spread, still assigning reasonably high probability mass to the neighbourhood around the true parameter value. 
Figure~\ref{fig:ma2_updated_posteriors} shows updated beliefs for the local effects, obtained using the update rule  of Equation~\eqref{eq:update_pi_j} in Section~\ref{sec:mba}. The updated beliefs exhibit shrinkage towards the global mean effect and, in this case, concentrate more accurately around the actual local effect values.
On the other hand, many of the REMA posteriors for the local effects, shown in Figure~\ref{fig:ma2_re_posteriors}, are biased and concentrate in regions away from true values.

\subsection{Example 2: Tuberculosis outbreak dynamics}
\label{subsec:tbc}

We now apply MBA to conduct meta-analysis of parameters regulating a stochastic birth-death (SBD) model proposed by \citet{Lintusaari+others:2019}, who used their model in a single-study setting to analyze tuberculosis outbreak data from the San Francisco Bay area, initially reported by \citet{Small+others:1994}. 
The goal of the analysis was to estimate disease transmission parameters from genotype data which, in contrast to outbreak models relying on count data, renders the likelihood function intractable and necessitates the use of likelihood-free inference \citep{Tanaka+others:2006}. 
Furthermore, such models are often complex in relation to the available data, which may result in poor identifiability, as discussed by \citet{Lintusaari+others:2016}. 
To alleviate the problem, they formulated their model as a mixture of stochastic processes, taking into account the individual transmission dynamics of different subpopulations. 
In our analyses, we focus on two key parameters of the model, $R_0^{(1)}$ and $R_0^{(2)}$, which are the reproductive numbers for two subpopulations: those that are compliant and non-compliant to treatment, respectively\footnote{Note that \citet{Lintusaari+others:2019} used the notation $R_0$ and $R_1$, respectively. The standard notation for the reproductive number is $R_0$}.

For our current experiment, we analyzed three additional data sets using the model of \citet{Lintusaari+others:2019}. These data sets reported tuberculosis outbreaks in Estonia \citep{Kruuner3339}, London \citep{Maguire617} and the Netherlands \citep{Netherlands}. 
For each data set, we independently conducted likelihood-free inference, generating $1000$ samples from the posteriors.
Following \citet{Lintusaari+others:2019}, we used the following six summary statistics for the ABC simulations: the number of observations, the total number of genotype clusters, the size of the largest cluster, the proportion of clusters of size two, the proportion of singleton clusters, and finally, the average successive difference in size among the four largest clusters. The original publication additionally used two summary statistics on the observation times of the largest cluster. While these improved model identifiability, such information was not available for the additional data sets analyzed in our current experiment.
We used a weighted Euclidean distance as dissimilarity function, with the same weights as \citet{Lintusaari+others:2019}.

Figure~\ref{fig:tb_abc_posteriors} shows the joint posterior distributions for the parameters $R_0^{(1)}$ and $R_0^{(2)}$, obtained individually for all four geographical locations using ABC.  
Compared to the San Francisco data, the posteriors computed on the remaining data sets, in particular London and the Netherlands, show symptoms of unidentifiability, with the posterior mass spread along lines between, and including, the boundaries of the parameter space. 
This can at least partly be attributed to these data sets being less informative than the San Francisco data set. A key question for our experiment then is whether we could borrow strength across the studies to improve the identifiability of the models computed on the remaining data sets. Additionally, it will be of interest to obtain an overall analysis of the central tendency of the reproductive numbers.

As in the previous experiment of Section~\ref{subsec:ma2}, we first define a model for the local effects $\theta_j = (R_{0,j}^{(1)}, R_{0,j}^{(2)})$, and assign priors for the components of $\varphi=(\mu, \Sigma_0)$, where $\mu = (\mu_1,\mu_2)$ is the global mean effect and $\Sigma_0$ is the covariance matrix. This results in the model:
\begin{gather*}
    \theta_j\sim \mathcal{N}_2(\mu, \Sigma_0), \quad
    \mu_1 \sim \mathrm{Gamma}(a_1, b_1),\quad
    \mu_2 \sim \mathrm{Gamma}(a_2, b_2), \quad
    \Sigma_0\sim \mathcal{W}^{-1}(\nu,\Psi). 
\end{gather*}
The hyperparameters were set as follows:
\[
a_1 = 0.12,\, b_1 = 0.36\quad 
a_2 =.030, \, b_2 = 0.05,\quad  \text{ and } \quad
\nu = 4,\quad  
\Psi = \begin{bmatrix}
4 & -0.1 \\
-0.1 & 0.01
\end{bmatrix}.
\]
We then follow the same procedure from the previous experiments to sample from the MBA-updated effects via SIR.

Note that due to the indirect nature in the relationship between the infection data and the parameters of interest, using the data to directly construct an estimator for the parameters of interest would be difficult.
While this does not pose a challenge in our framework, it renders the application of traditional meta-analysis approaches infeasible.

Figure~\ref{fig:tb_abalation} shows the updated beliefs for the local effects $\theta_1, \ldots, \theta_J$, after borrowing strength across the individual studies. While the updated beliefs clearly retain some of the individual characteristics of their original counterparts (e.g. a similar covariance structure), they exhibit a much more identifiable behavior. The superposterior of the overall mean of the reproductive numbers is shown in 
Figure~\ref{fig:tb_superposterior}.

\begin{figure}[H]
\centering
\includegraphics[width=0.85\linewidth]{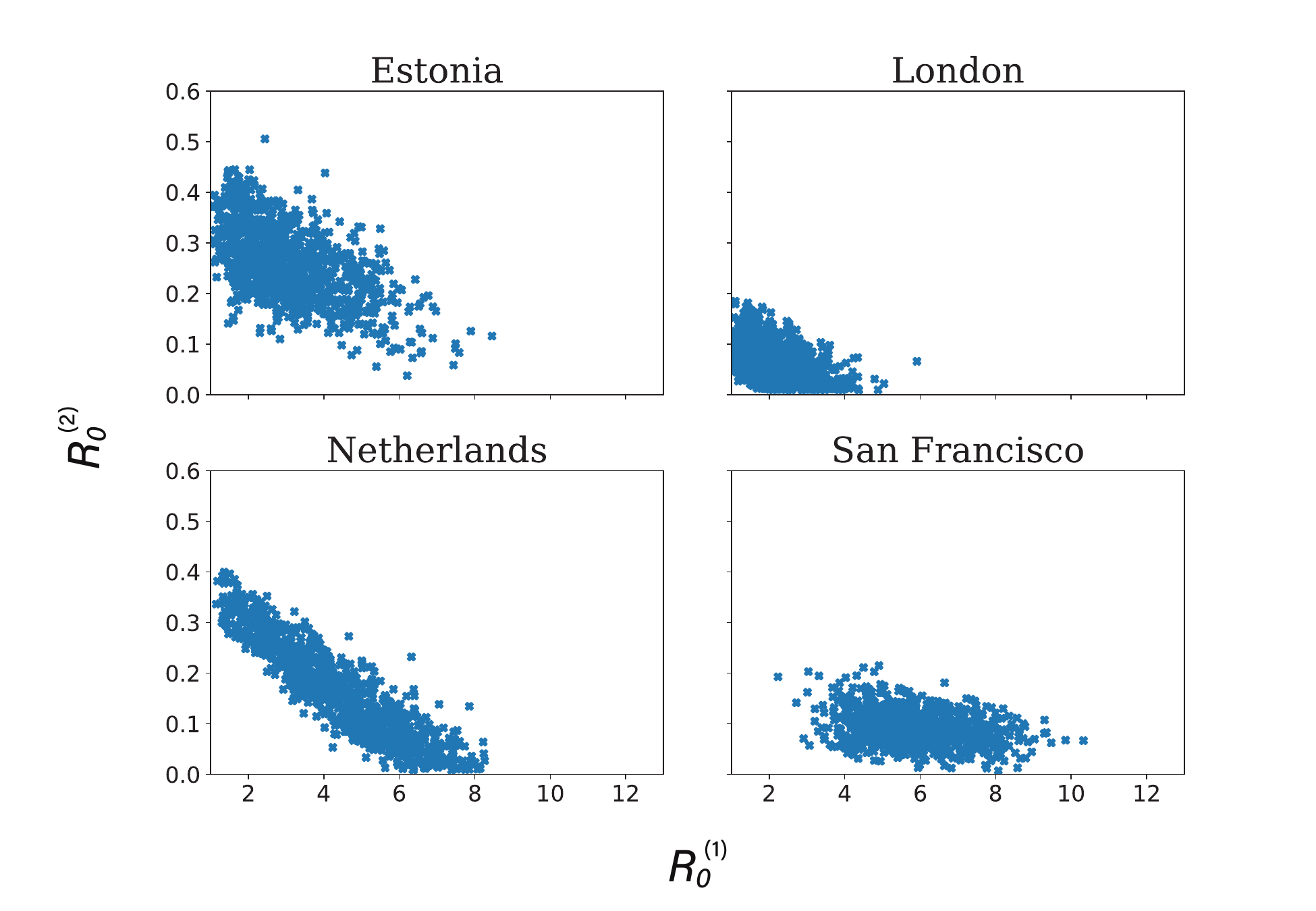}
\\
\caption{Posteriors for the reproductive numbers $(R_0^{(1)}, R_0^{(2)})$, individually obtained using ABC in four different studies on tuberculosis outbreak dynamics.}
\label{fig:tb_abc_posteriors}
\end{figure}

\begin{figure}[H]
\centering
\includegraphics[width=0.85\linewidth]{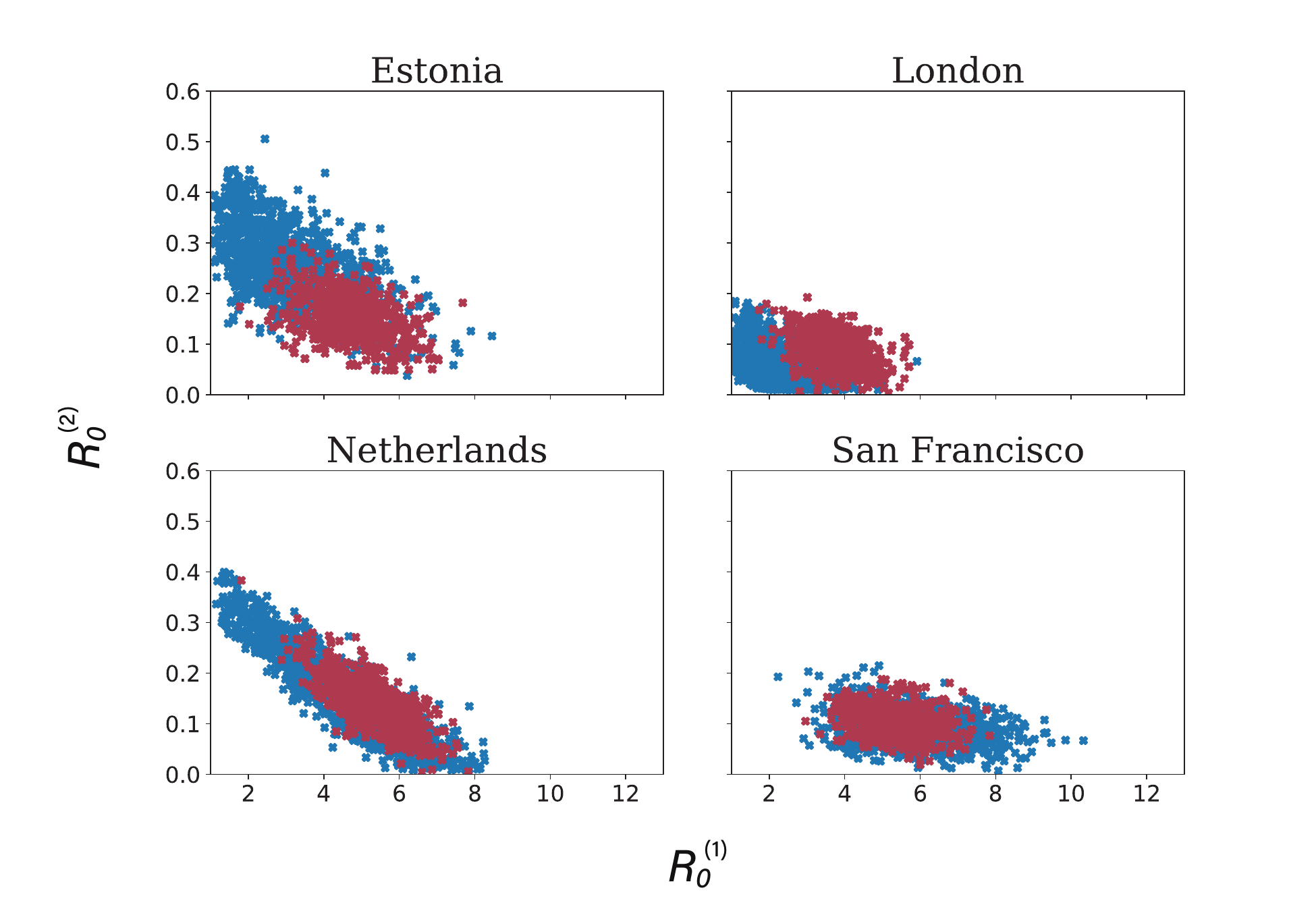}
\\
\caption{Posteriors for the reproductive numbers $(R_0^{(1)}, R_0^{(2)})$ updated using MBA (red), plotted on top of the original, individually obtained posteriors (blue).}
\label{fig:tb_abalation}
\end{figure}

\begin{figure}[H]
\centering
\includegraphics[width=0.42\linewidth]{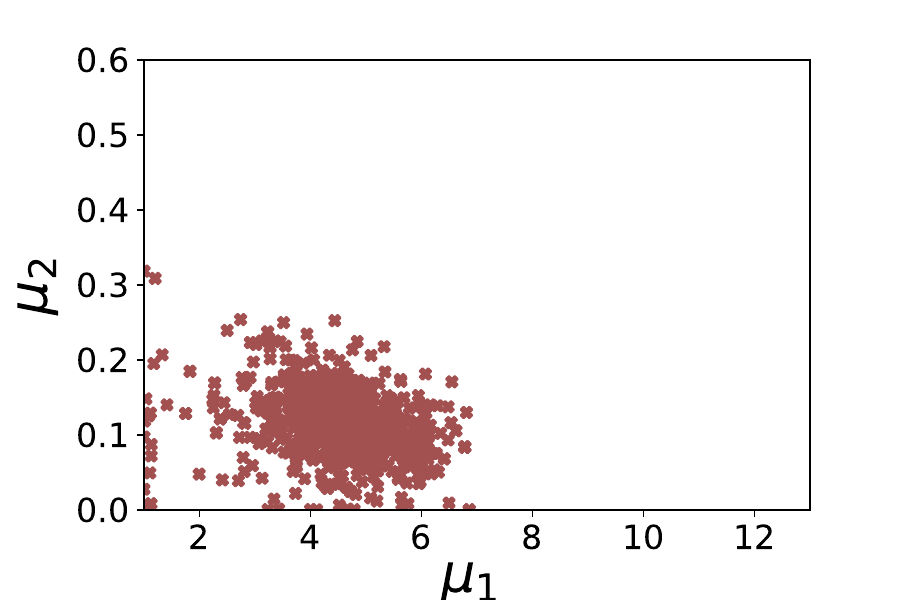}\\
\caption{MBA joint superposterior for the overall mean effect $\mu=(\mu_1, \mu_2)$. }
\label{fig:tb_superposterior}
\end{figure}

In the supplementary material (Appendix E) we provide a synthetic version of these tuberculosis experiments, again showing that sharing statistical strength across studies yields local posteriors with improved identifiability compared to their original counterparts.

\section{Bayesian inference with distributional observations}\label{sec:theory}

Overall, our approach to conduct meta-analysis from a set of related posteriors, MBA, can be seen more broadly as a framework for Bayesian inference with distributional observations.
This section develops a foundation for this framework, which serves as the theoretical justification for MBA. 
As we will see, this framework is very general and encompasses the standard settings of point-valued observations as special cases, given the appropriate restrictions.

We start our presentation by treating $\theta_{1}, \ldots, \theta_{J}$ as directly observable, point-valued random quantities. 
Then, we modify Bayes' rule to incorporate distributions $\Pi_1,\ldots, \Pi_J$ over $\theta_1, \ldots, \theta_J$ by substituting the conventional likelihood factors
with their expected values, arriving at the measure-theoretic version of Equation~\eqref{eq:overall_effect}.

\noindent\textbf{Remark:}
We switch notations slightly in this section: Instead of the probability density $q(\varphi)$ for $\varphi$ we shall use the probability measure $Q(B)$ for sets $B$ of possible values of $\varphi$; and accordingly, the measure-theoretic style shall be used for integrals. The fact that $Q(B) = 
\int_{B} Q(d\varphi) = \int_{B} q(\varphi)d\varphi$ may help readers to switch back and forth between the measure-theoretic notations and their counterparts in terms of densities, when they exist. 
A similar remark applies to all other variables and parameters.

To gain intuition, 
let us first suppose that $\theta_1,\ldots,\theta_J$ are observable and exchangeable random quantities. 
Following standard theory, de~Finetti's representation theorem \citep[e.g.][]{Schervish:1995} states that if 
$\theta_1,\theta_2,\ldots$ is an infinitely exchangeable sequence of random
quantities taking values in a Borel space $(\Theta$, $\mathcal{A})$,  
then there exists a probability measure $Q$ on a Borel space $(\Phi,\mathcal{B})$
such that the joint distribution of the subsequence $\theta_1, \ldots,\theta_J$, i.e. the \emph{predictive distribution}, has the `mixture' form
\begin{equation}\label{eq:parametric_deFinetti_measures}
\P(\theta_1 \in A_1, \ldots,\theta_J \in A_J) = \int_{\Phi} \prod_{j=1}^J P_{\varphi}(A_j) \ Q(d\varphi).
\end{equation}
It is understood that the equation holds for any sets $A_1,\ldots,A_J \in\mathcal{A}$. The details of the theorem are omitted for simplicity, but the essence of de Finetti's representation is establishing that exchangeability of $\theta_1,\theta_2,\ldots$ is equivalent to conditional independence of these quantities given the value $\varphi$ of a random parameter on a suitable space $\Phi$. The family $\{P_{\varphi}|\varphi\in\Phi\}$ consists of probability measures on $\Theta$, indexed by the parameter $\varphi$. 
Assuming that these measures have density functions $p(\cdot|\varphi):= dP_{\varphi}/d\lambda$ 
with respect to a reference (Lebesgue or counting) measure $\lambda$, and $Q$ has density $q$, one may rewrite  Equation~\eqref{eq:parametric_deFinetti_measures} in the following (arguably friendlier) form:

\begin{equation}\label{eq:parametric_deFinetti_densities}
\P(\theta_1 \in A_1, \ldots,\theta_J \in A_J) = \int_{\Phi} \prod_{j=1}^J \left[\int_{A_j}p(\theta_j|\varphi)\lambda(d \theta_j) \right] 
q(\varphi) d\varphi.
\tag{\ref*{eq:parametric_deFinetti_measures}'}
\end{equation}

For reference, the corresponding joint density, if it exists, then is as follows:
\begin{equation}
\label{eq:parametric_deFinetti_joint_density}
p(\theta_1, \ldots,\theta_J) = \int_{\Phi} \prod_{j=1}^J p(\theta_j|\varphi) 
q(\varphi) d\varphi.
\end{equation}

In the above standard setting, one may think of $Q$ as being a prior distribution on the parameter $\varphi$, and the Bayesian learning process works through updating $Q$ conditional on observed data. Following the above equations, the posterior distribution of $\varphi$, given observed values $\theta_{1}=t_1, \ldots,\theta_J=t_J$, has the form
\begin{equation}\label{eq:posterior}
Q(B|t_1,\ldots,t_J) = \frac{\int_{B} \prod_{j=1}^J p(t_j|\varphi) \, Q(d\varphi)}{\int_{\Phi} \prod_{j=1}^J p(t_j|\varphi) \, Q(d\varphi)}.
\end{equation}
The equation holds for any set $B\in\mathcal{B}$, where $\mathcal{B}$ is the Borel $\sigma$-algebra on $\Phi$. 
Let us emphasize that, as a preparatory step to build intuition, we have up until this point in the current section assumed that  the $\theta_{j}$ are directly observable, point-valued random quantities. 
This is in contrast with standard notational conventions, where $\theta$ usually denotes an unobservable parameter (cf. the hierarchical meta-analysis model described in Section~\ref{sec:intro}, where the point-valued observables are the summary statistics $D_j$).

We now consider a setting where the observables are distributions $\Pi_1,\ldots,\Pi_J$ on the values of $\theta_1,\ldots,\theta_J$ respectively. 
This setting therefore departs from the standard setting built for point-valued observations, and requires a different formulation.
We propose a simple modification as the basis to build an update rule for the new setting, substituting the likelihood factors $p(t_j|\varphi)$ of the standard setting with the \emph{expected likelihood factors}
$\int_{\Theta}p(\theta_j|\varphi)\Pi_j(d\theta_j)$, where the expectation is with respect to the observed distributions.
We argue that these factors are reasonable for the setting we are considering, where the observables are distributions; not only because they are computable on the basis of the observed distributions, but also because they lead to an update rule that retains some basic properties of the classical Bayesian update rules, see Section~\ref{sec:theoretical_properties}.

The proposed expected likelihood factors now lead to an update of the form
\begin{equation}\label{eq:posterior_uncertain_obs}
Q^*(B|\Pi_1,\ldots,\Pi_J) = \frac{\int_{B} \prod_{j=1}^J \left[\int_{\Theta}p(\theta_j|\varphi)\Pi_j(d\theta_j)\right] Q(d\varphi)}
{\int_{\Phi} \prod_{j=1}^J \left[ \int_{\Theta}p(\theta_j|\varphi)\Pi_j(d\theta_j)\right] Q(d\varphi)},
\end{equation}
for $B\in\mathcal{B}$,
where we write $Q^*(\cdot|\Pi_1,\ldots,\Pi_J)$ to denote conditioning on distributions, in analogy with conditioning on fully observed point values. We give more context for this choice of notation below in Section~\ref{sec:posterior_concentration}. 
We refer to the measure defined by Equation~\eqref{eq:posterior_uncertain_obs} as a \emph{superposterior distribution} to distinguish it from Equation~\eqref{eq:posterior} and to convey the notion of it being the result of an update on top of probability distributions $\Pi_j$.
It is easy to see that Equation~\eqref{eq:posterior} emerges as a special case of the superposterior by setting $\Pi_j$ to be $\delta_{t_j}$, the Dirac measure centered at $t_j$. This yields 
\[
\int_{\Theta} p(\theta_j|\varphi)\delta_{t_j}(d\theta_j) = p(t_j|\varphi),
\]
such that $Q^*(\cdot|\delta_{t_1},\ldots,\delta_{t_J}) = Q(\cdot|t_1,\ldots,t_J)$. 
Equation~\eqref{eq:posterior_uncertain_obs} also induces a joint distribution on the space $\Phi\times\Theta^J$, which we can marginalize with respect to $Q$ to obtain a predictive distribution:
\begin{equation}\label{eq:joint_belief}
\P^*(\theta_1 \in A_1, \ldots,\theta_J \in A_J) = \frac{\int_{\Phi} \prod_{j=1}^J \left[\int_{A_j}p(\theta_j|\varphi)\Pi_j(d\theta_j)\right] Q(d\varphi)}
{\int_{\Phi} \prod_{j=1}^J\left[ \int_{\Theta}p(\theta_j|\varphi)\Pi_{j}(d\theta_j)\right] Q(d\varphi)}.
\end{equation}

Throughout this work, we assumed that every $\Pi_j$ is a probability measure.
It is interesting to note, however, that if we allow $\Pi_j$ to be the Lebesgue (or counting) measure $\lambda$ for all $j$, which corresponds to having a uniformly distributed---possibly improper---belief about the value of $\theta_j$, then the updated measure $Q^*(\cdot|\Pi_1,\ldots,\Pi_J)$ in Equation~\eqref{eq:posterior_uncertain_obs} equals the prior probability measure $Q$. 
Moreover, with this choice of $\Pi_j$, Equation~\eqref{eq:joint_belief} reduces to the standard predictive distribution in Equation~\eqref{eq:parametric_deFinetti_measures}.

\subsection{Theoretical properties}\label{sec:theoretical_properties}

Two well-known properties of standard Bayesian inference, which are of practical relevance in our meta-analysis setting, are: (i) order-invariance in successive posterior updates, and (ii) posterior concentration. The former ensures that inferences conditional on exchangeable data are coherent. The latter tells us that the posterior distribution becomes increasingly informative about the quantity of interest, as we accumulate more observations.
We now briefly discuss these properties in the context of the framework introduced above.
Formal proofs are given in Appendix A and Appendix B.

\subsubsection{Order-invariance in successive updates}\label{sec:order_invariance}

Under the exchangeability assumption, standard Bayesian inference can be constructed as a sequence of successive updates, invariant to the order in which data are processed. It is straightforward to verify that the superposterior update defined in Equation~\eqref{eq:posterior_uncertain_obs} shares this property of being order-invariant. The proof of the following proposition is given in the supplementary material (Appendix A).
\begin{proposition}
The superposterior update defined in Equation~\eqref{eq:posterior_uncertain_obs} is 
invariant to permutations of the indices $1,\ldots,J$.
\end{proposition}

\subsubsection{Concentration in the limit of infinite observations}
\label{sec:posterior_concentration}

Asymptotic theory states that if a consistent estimator of the true value (or an optimal one in terms of KL divergence) of the parameter $\varphi$ exists, then the posterior distribution in Equation~\eqref{eq:posterior} concentrates in a neighborhood of this value, as $J\rightarrow\infty$ \citep[e.g.][]{Schervish:1995}. 
Here we discuss conditions under which the same property holds for the measure $Q^*(\cdot|\Pi_1,\ldots,\Pi_J)$, defined in Equation~\eqref{eq:posterior_uncertain_obs}.

Our strategy is to first formulate a conceptual model for the observed distributions $\Pi_1,\ldots,\Pi_J \in \mathcal{P}(\Theta)$; where the notation $\mathcal{P}(\Theta)$ stands for the space of probability measures on $\Theta$.
To this end, consider the following:
\begin{subequations}
\begin{align}
\varphi &\sim Q\label{eq:hierarchical_uncertain_obs_Q}\\
\Pi_j &\sim H^{(j)}_{\varphi}.
\label{eq:hierarchical_uncertain_obs_G}
\end{align}
\end{subequations}
Thus, $H^{(j)}_{\varphi}$ is the distribution of $\Pi_j$, given $\varphi$. We assume that these distributions are defined by their densities $h_j(\cdot|\varphi):= dH^{(j)}_{\varphi}/d\kappa$ with respect to a dominating measure $\kappa$ on $\mathcal{P}(\Theta)$. 
Furthermore, assuming integrability, we take $h_j(\Pi_j|\varphi) := C_j \int_{\Theta}p(\theta_j|\varphi) \Pi_j(d\theta_j)$, where the $C_j$ are the normalizing factors such that $h_j$ integrates to 1. 
Then the probability measure $Q^*(B|\Pi_1,\ldots,\Pi_J)$ defined in Equation~\eqref{eq:posterior_uncertain_obs} takes the form 
\begin{equation}\label{eq:posterior_uncertain_obs_v3}
Q^*(B|\Pi_1,\ldots,\Pi_J) = \frac{\int_{B} \prod_{j=1}^J h_j(\Pi_j|\varphi) \, Q(d\varphi)}{\int_{\Phi} \prod_{j=1}^J h_j(\Pi_j|\varphi) \, Q(d\varphi)}.
\end{equation}

An essential difference between the standard posterior distribution in Equation~\eqref{eq:posterior} and that of  Equation~\eqref{eq:posterior_uncertain_obs_v3} is that, in the former, the observations are conditionally i.i.d., while in the latter they are conditionally independent but \emph{non-identically} distributed. 
Note that the observed distributions $\Pi_j$ might in general be the outputs of Bayesian analyses with different prior assumptions and design choices  
across the various studies. 
Furthermore, in our setting, the meta-analyst may not have access to the exact computational details of each study, and indeed, only access to the posterior distribution produced by each study is required. Hence, it is not reasonable to postulate homogeneity across the various studies.
Nonetheless, from a practical point of view, we assume that the 
posteriors in each study target the same set of parameters. 
Another important difference to highlight is that Equation~\eqref{eq:posterior_uncertain_obs_v3} is based on  \emph{distributional} observations, as our setting considers measure-valued observations as the inputs, whereas Equation~\eqref{eq:posterior} is based on point-valued observations, as per the standard setting.

One can interpret Equation~\eqref{eq:posterior_uncertain_obs_v3} as giving a representation of $Q^*(B|\Pi_1,\ldots,\Pi_J)$ which is sufficient to establish concentration of this distribution in the limit $J\to\infty$. 
These results are given in the supplementary material (Appendix B, 
where Section B.1 covers discrete parameter spaces, and  Section B.2 covers general parameter spaces).

For the case of discrete parameter spaces, the proof follows the basic structure found in many sources; interested readers are referred to \citet[App.~B, p.~588]{Gelman+others:2013} and \citet[Sec.~5.3, p.~286]{Bernardo+Smith:1994}.
Note however that these references assumed i.i.d. point-valued observations; whereas our setting deals with measure-valued observations which are independent but non-identically distributed.
A key step in proving the concentration result is to establish a limit for the sums of log-likelihood ratios $\sum_{j=1}^J\log (h_j(\Pi_j|\varphi_0)/h_j(\Pi_j|\varphi))$ for $\varphi\neq\varphi_0$, where $\varphi_0$ is a distinguishable parameter value which meets some technical requirements. 
This limit is established using a version of the strong law of large numbers that is valid for non-identically distributed random variables \citep[cf.][Theorem 2.3.10]{Sen+Singer:1993}. 
For continuous parameter spaces the argument is more intricate and requires the technical assumptions to be adjusted in some ways; however, the basic building blocks of the proof are analogous to those of the discrete case.
Interested readers are referred to 
\citet[Theorem~7.80]{Schervish:1995} for the result for continuous parameter space
and i.i.d. point-valued observations. Our result extends this theorem to the setting of observations which are measure-valued and independent but non-identically distributed.

In conclusion, if the superposterior distribution $Q^*(\cdot|\Pi_1,\ldots,\Pi_J)$ can be written in the form of Equation~\eqref{eq:posterior_uncertain_obs_v3} and some technical conditions are met, then concentration when $J\to\infty$ can be established,
using similar arguments to those that establish the concentration results for classical posterior distributions.

\section{Related work}\label{sec:related_work}

To the best of our knowledge, meta-analysis frameworks for combining posterior distributions have not been presented before. 
Nonetheless, similar elements can be found in previous works.

\subsection{Bayesian meta-analysis} 
In the context of conventional Bayesian meta-analysis, \citet{Lunn+others:2013} introduced a two-stage computational approach for fitting hierarchical models to individual-level data. 
In the first stage, posterior samples for local parameters are generated independently for each study, enabling study-specific complexities to be dealt with separately. In a second stage, the samples are used as proposals in Metropolis-Hastings updates within a Gibbs algorithm to fit the full hierarchical model.
While the paper focuses on a specific computational approach for hierarchical models, involving no propagation of local prior knowledge into the joint model, the idea of utilizing independently computed posterior samples to fit a joint model is common to our framework. 

In more recent work, \citet{Rodrigues+others:2016} developed a hierarchical Gaussian process prior to model a set of related density functions, where grouped data in the form of samples assumed to be drawn under each density function are available.
Despite a superficial similarity between their work and ours, the inferential goals in these works are very different. Specifically, the former is concerned with nonparametric estimation of group-specific densities, which is shown to be useful when the sample sizes in some or all of the groups are small.  
Thus, in cases where the number of posterior samples per study is limited (e.g. due to computational reasons), the method of \citet{Rodrigues+others:2016} could be used within our framework to provide density estimates for the initial beliefs. %

\subsection{Frequentist meta-analysis}
In the frequentist paradigm, \citet{Xie+others:2011} introduced a general framework for the combination of confidence distributions, a concept loosely related to Bayesian posteriors and deeply rooted in the frenquentist paradigm \citep[see][]{Schweder_Hjort:2002}. 
In the same realm, \citet{Liu2015} combine analyses from related results using a product of confidence distributions, under the assumption that known linear functions map the local effects to the overall effect.

\subsection{Bayesian updating using uncertain evidence}

Our work builds on the interpretation that each $\pi_j(\theta_j)$ is a probabilistic representation of belief \citep{Bernardo+Smith:1994}, which reflects uncertainty about the value of the corresponding local effect $\theta_j$.  
Updating prior knowledge (in our case regarding $\varphi$) subject to uncertain or `soft' evidence, represented as probability distributions on the local effects instead of point estimates, has been extensively studied as a philosophical topic in both statistics and artificial intelligence. 
The most well-known of such update rules, 
Jeffrey's rule of conditioning \citep{Diaconis+Zabell:1982,Jeffrey:2004,Smets:1993,Zhao+Osherson:2010}, 
computes the updated probability for an event as a weighted average of the posterior probabilities under all possible values of the evidence.
Due to its construction, 
Jeffrey's rule is applicable in simple discrete cases but becomes computationally infeasible for more complex models with continuous variables.
It is also well known that Jeffrey's update rule is in general not invariant to the order in which data are processed, unlike our proposed update, as discussed in Section~\ref{sec:order_invariance}. 

\subsection{Distributed Bayesian inference}
We finally note that while our framework can be described as an approach to combine (information provided by) a set of pre-computed posteriors, it is substantially different from works in the field of \emph{Embarrassing Parallel} (or divide-and-conquer) inference~\citep[e.g.][]{Neiswanger+Xing:2017, Wang+others:2018}. 
In these works, given a partitioned dataset $\mathcal{X} = \cup_{j=1}^{J} \mathcal{X}_j$, a posterior $p(\theta|\mathcal{X})$ is approximated by first independently computing a \emph{subposterior} $p(\mathcal{X}_j|\theta) p(\theta)^{1/J}$ for each data partition, and then combining these as $p(\theta|\mathcal{X}) \propto \prod_{j=1}^{J} p(\mathcal{X}_j|\theta) p(\theta)^{1/J}$.
In contrast, the goal of our framework is to use posteriors computed over \emph{distinct} quantities $\theta_1,\ldots,\theta_J$ related through exchangeability, to provide information about a shared global quantity $\varphi$.

\section{Conclusion and future work}\label{sec:discussion}

The natural outcome of a Bayesian analysis is a posterior distribution over quantities of interest.
Meta-analysis methods combine results from multiple related statistical studies into a consensus analysis,
and to the best of our knowledge all previous approaches to conduct meta-analyses have used summary statistics (or individual-level data) as the observable inputs.
In this paper, 
we have developed a framework that uses distributions as the observable inputs to conduct meta-analyses, thus enabling the combination of posteriors from multiple related Bayesian analyses.
The framework builds on standard Bayesian inference over exchangeable quantities, by treating each observed posterior as data observed with uncertainty. 
We have also shown how the framework can be used to update independently computed, study-specific posteriors \emph{post-hoc} through sharing of statistical strength. This is analogous to how estimates for related parameters share strength between each other in hierarchical models \citep[e.g.][]{Gelman+others:2013}.

In many fields, it has become common practice to store data sets in dedicated repositories to be reused for the benefit of the entire research community. 
Given the view taken in this work, that posterior distributions can be seen as a special kind of observational data, we believe that in many cases it would be equally beneficial to make full posterior distributions available for reuse. 
This would enable posteriors from potentially time-consuming and costly Bayesian analyses to be used as a basis for new studies.
Indeed, even if the original data, the model and the code implementing it were available, reproducing posterior distributions could require a substantial computational effort.
While in this work we have used posterior samples as a representation for study-specific posteriors, the optimal format for storing and sharing posterior distributions is a question for future research.
In addition to ways of making posteriors publicly available, more research is needed on developing methods to make appropriate use of the information they provide.
The current work is a first step in this direction and our hope is that it will inspire other researchers to make further advances to this end.

Furthermore, we believe exploring MBA as a framework for Bayesian inference using uncertain evidence is a promising direction for future work. For instance, in the context of multiple imputations,  MBA could be used to directly incorporate the posterior predictive distributions for the missing values as beliefs, avoiding repeated MCMC runs for multiple imputed datasets~\citep{Missing}. Given that MBA accommodates the combination of posteriors and point estimates in a single meta-analysis, another relevant direction for future work is to extend the analysis to these settings; however further research would be required for a careful study of the properties of such hybrid  combinations of posteriors and point estimates.

\section*{Acknowledgments}

This work was supported by the Research Council of Finland (Flagship programme: Finnish Center for Artificial Intelligence FCAI, 359207), EU Horizon 2020 (European Network of AI Excellence Centres ELISE, grant agreement 951847), UKRI Turing AI World-Leading Researcher Fellowship (EP/W002973/1), Fundação Carlos Chagas Filho de Amparo à Pesquisa do Estado do Rio de Janeiro FAPERJ (SEI-260003/000709/2023), São Paulo Research Foundation FAPESP (2023/00815-6), Conselho Nacional de Desenvolvimento Científico e Tecnológico CNPq (404336/2023-0). We also acknowledge the computational resources provided by the Aalto Science-IT Project from Computer Science IT.

\bibliographystyle{chicago}
\bibliography{BA_biblia.bib}



\renewcommand{\theHsection}{A\arabic{section}}

\newpage

\makeatletter
\providecommand{\maketitleappendix}{}
\renewcommand{\maketitleappendix}{%
    \title{\LARGE\textbf\textsf
    Supplementary Material 
    }
}
\maketitleappendix
\appendix

\section{Proof of Proposition 2.1}
\label{sec:proof_order_invariance}

It suffices for us to verify the claim for $J=2$. Beginning with $J=1$, we update the probability $Q(B)$ into $Q^*(B|\Pi_1)$ using Equation~\eqref{eq:posterior_uncertain_obs} in the main text: 
\[
Q^*(B|\Pi_1) = \frac{\int_{B} \int_{\Theta}p(\theta_1|\varphi)\Pi_1(d\theta_1) Q(d\varphi)}
{\int_{\Phi} \int_{\Theta}p(\theta_1|\varphi)\Pi_1(d\theta_1) Q(d\varphi)}.
\]
Notice that $Q^*(\cdot|\Pi_1)$ is absolutely continuous with respect to $Q$ and the Radon-Nikodym derivative is given by 
$
\frac{dQ^*(\cdot|\Pi_1)}{dQ}(\varphi)
= C\int_{\Theta}p(\theta_1|\varphi)\Pi_1(d\theta_1)
$
with normalizing factor $C >0$.

Then, we reapply Equation~\eqref{eq:posterior_uncertain_obs} to update $Q^*(B|\Pi_1)$ into $Q^*(B|\Pi_1,\Pi_2)$:
\begin{align*}
Q^*(B|\Pi_1,\Pi_2) 
&= \frac{\int_{B} \int_{\Theta}p(\theta_2|\varphi)\Pi_2(d\theta_2 ) Q^*(d\varphi|\Pi_1) }
{\int_{\Phi} \int_{\Theta}p(\theta_2|\varphi)\Pi_2(d\theta_2) Q^*(d\varphi|\Pi_1)}\\
&= \frac{
\int_{B} \int_{\Theta}p(\theta_2|\varphi)\Pi_2(d\theta_2 ) \frac{dQ^*(\cdot|\Pi_1)}{dQ}(\varphi)\cdot Q(d\varphi)
}
{
\int_{\Phi} \int_{\Theta}p(\theta_2|\varphi)\Pi_2(d\theta_2) \frac{dQ^*(\cdot|\Pi_1)}{dQ}(\varphi)\cdot Q(d\varphi)
}
\\
&= \frac{\int_{B} \prod_{j=1}^2 \left[\int_{\Theta}p(\theta_j|\varphi)\Pi_j(d\theta_j)\right] Q(d\varphi)}
{\int_{\Phi} \prod_{j=1}^2 \left[ \int_{\Theta}p(\theta_j|\varphi)\Pi_j(d\theta_j)\right] Q(d\varphi)},
\end{align*}
which is equivalent to a direct application of the equation for $J=2$, and independent of the order in which $\Pi_1$ and $\Pi_2$ are processed.

\section{Proof of posterior concentration}
\label{sec:proof_pc}

We quote a strong law of large numbers which is valid for independent but not necessarily identically distributed observations \citep[cf.][Theorem 2.3.10]{Sen+Singer:1993}:
\begin{theorem}
\label{thm:sen_singer_appendix}
Assume that the random variables $(\xi_j)_{j\geq 1}$ are independent. Assume further that 
$\E(\xi_j)=\mu_j$ and $\operatorname{Var}(\xi_j)=\sigma_j^2$ exist for all $j\geq 1$, and $\sum_{j\geq 1} j^{-2} \sigma_j^2 < \infty$.
Then in the limit when $J \to \infty$ we have:
\[
\frac{1}{J}\sum_{j=1}^J \xi_j - \frac{1}{J}\sum_{j=1}^J \mu_j \xrightarrow{\text{a.s.}} 0.
\]
\end{theorem}

We recall the definition of the Kullback-Leibler functional, also called KL divergence. For probability distributions $\nu_0, \nu$ on the same space, 
\[
\KL(\nu_0;\nu) 
= \int \log\left( \frac{d\nu_0}{d\nu} \right) d\nu_0
\]
whenever $\nu_0$ is absolutely continuous with respect to $\nu$, and $\KL(\nu_0;\nu)=\infty$ otherwise.

\subsection{Discrete parameter space}\label{sec:discrete_space}

Here, we give an elementary proof of concentration of the superposterior distribution $Q^*(\cdot|\Pi_1,\ldots,\Pi_J)$ for discrete parameter spaces, in the limit of infinite observations.   
The proof follows the basic structure found in many sources, interested readers are referred to \citet[App.~B, p.~588]{Gelman+others:2013} where the basic proof is outlined for discrete spaces and i.i.d. observations, see also \citet[Sec.~5.3, p.~286]{Bernardo+Smith:1994}.

Assuming that there exists, for all $j$, a unique minimizer $\varphi_0$ of the KL divergence from the true distribution of $\Pi_j$ to the parametrized representation $H_\varphi^{(j)}$ as in Section~\ref{sec:posterior_concentration}, a key step in proving the said concentration is to establish a limit for the sums of the log-likelihood ratios $\sum_{j=1}^J\log(h_j(\Pi_j|\varphi_0)/h_j(\Pi_j|\varphi))$ for $\varphi\neq\varphi_0$. Notice that in this setting the summands are independent but non-identically distributed random variables, hence Theorem~\ref{thm:sen_singer_appendix} may be applied to obtain this limit, provided that its conditions are satisfied.

\begin{theorem}
\label{thm:concentration_discrete}
Assume the model (\ref{eq:hierarchical_uncertain_obs_Q})--(\ref{eq:hierarchical_uncertain_obs_G}).
Let $\Pi_1,\ldots,\Pi_J$ be distributional observations, and let $\bigl(\{H^{(j)}_{\varphi}|\varphi\in\Phi\}\bigr)_{j=1}^J$ be the corresponding family of marginals. 
Suppose that 
\begin{itemize}
\item[(i)] $\Phi$ is a countable set (possibly finite).
\end{itemize}
Furthermore, suppose $\varphi_0\in\Phi$ is a distinguishable parameter value in the sense of satisfying
\begin{itemize}
\item[(ii)] $Q(\varphi_0)>0$, and 
\item[(iii)]$\min_{\varphi\in\Phi\setminus\{\varphi_0\}} \KL\bigl(H^{(j)}_{\varphi_0};H^{(j)}_{\varphi}\bigr) > 0$, for all $j$.
\end{itemize}
Then, provided that the log-likelihood ratio terms $\xi_j:=\log\frac{h_j(\Pi_j|\varphi_0)}{h_j(\Pi_j|\varphi)}$ for $\varphi\neq\varphi_0$ satisfy the conditions of Theorem~\ref{thm:sen_singer_appendix}, the superposterior satisfies $Q^*(\varphi_0|\Pi_1\ldots,\Pi_J)\rightarrow 1$ as $J\rightarrow\infty$.
\end{theorem}

\begin{proof}
For any $\varphi\neq\varphi_0$, the log posterior odds can written as
\begin{equation}\label{eq:log_posterior_odds}
\log\frac{Q^*\left(\varphi_0|\Pi_1,\ldots,\Pi_J\right)}{Q^*\left(\varphi|\Pi_1,\ldots,\Pi_J\right)} 
= \log\frac{Q(\varphi_0)}{Q(\varphi)} + 
\sum_{j=1}^J\log \frac{h_j(\Pi_j|\varphi_0)}{h_j(\Pi_j|\varphi)},
\end{equation}
where the second term is a sum of $J$ independent but non-identically distributed random variables. By Theorem~\ref{thm:sen_singer_appendix}, we have that 
\[
\frac{1}{J}\sum_{j=1}^J\log \frac{h_j(\Pi_j|\varphi_0)}{h_j(\Pi_j|\varphi)}
- \frac{1}{J}\sum_{j=1}^J \int \left(\log \frac{h_j(\Pi_j|\varphi_0)}{h_j(\Pi_j|\varphi)}\right) H^{(j)}_{\varphi_0}(d\Pi_j)
\longrightarrow 0
\]
with probability 1, as $J\rightarrow\infty$. 
By condition (iii) we have
\[
\frac{1}{J}\sum_{j=1}^J \int \left(\log \frac{h_j(\Pi_j|\varphi_0)}{h_j(\Pi_j|\varphi)}\right) H^{(j)}_{\varphi_0}(d\Pi_j) 
\geq \min_{\varphi\in\Phi\setminus\{\varphi_0\}} \KL\bigl(H^{(j)}_{\varphi_0};H^{(j)}_{\varphi}\bigr) > 0,
\]
and consequently,
\[
\sum_{j=1}^J\log \frac{h_j(\Pi_j|\varphi_0)}{h_j(\Pi_j|\varphi)}
\longrightarrow \infty.
\]
Since $Q(\varphi_0)>0$, the entire expression (\ref{eq:log_posterior_odds}) approaches $\infty$ as $J\rightarrow \infty$, which implies that $Q^*(\varphi|\Pi_1,\ldots,\Pi_J)\rightarrow 0$ for $\varphi\neq\varphi_0$ and, since $\sum_{\varphi\in\Phi}Q(\varphi)=1$ ($Q$ is a probability), necessarily $Q^*(\varphi_0|\Pi_1,\ldots,\Pi_J)\rightarrow 1$.
\end{proof}

The interested readers may take note that this proof essentially imitated the one outlined in \citet[App.~B]{Gelman+others:2013}, 
with suitable adaptations to our setting.

\subsection{Continuous parameter space}\label{sec:continuous_space}

The space $\Theta$ encompasses the possible values of the effects of interest, and throughout the remaining of this section $\mathcal{P}(\Theta)$ is the space of probability measures on $\Theta$.

Recall (as per de Finetti's theorem) that
$Q$ is a probability measure on a Borel space $(\Phi,\mathcal{B})$
such that the joint distribution of $\theta_1, \ldots,\theta_J$ has the form given in Equation~\eqref{eq:parametric_deFinetti_measures}.

Fix $\varphi_0\in\Phi$, and for any $n \geq 1$ we define
\[
I_{n}(\varphi_0;\varphi) 
= \frac{1}{n} \sum_{j=1}^{n} \KL(H^{(j)}_{\varphi_0};H^{(j)}_{\varphi}).
\]

For any $\varepsilon>0$ and positive integers $n$, let $D_{n,\varepsilon}(\varphi_0) := \{ \varphi\in\Phi \ \vert \ I_{n}(\varphi_0;\varphi) < \varepsilon \}$, and define
$D_{\varepsilon}(\varphi_0) := \bigcap_{n\geq 1} D_{n,\varepsilon}(\varphi_0)$.

For any Borel set $A \subset \Phi$ define
\[
Z(A,\Pi_j) 
= \inf_{\psi \in A} \log\left( \frac{h_j(\Pi_j|\varphi_0)}{h_j(\Pi_j|\psi)} \right).
\]

We shall write $Q(D)$ for the $Q$-measure of sets $D$ in the corresponding space, and $H^{(j)}_{\varphi}[Z]$ for the expectation of random variables $Z$ with respect to the distribution $H^{(j)}_{\varphi}$.

After these technical definitions, we state and prove our concentration result, which builds on \citet[Theorem~7.80]{Schervish:1995}, with suitable adaptations to our setting.

\begin{theorem}
\label{thm:concentration_continuous}
Assume the model (\ref{eq:hierarchical_uncertain_obs_Q})--(\ref{eq:hierarchical_uncertain_obs_G}).
Let $\Pi_1,\ldots,\Pi_J$ be distributional observations, and let $\bigl(\{H^{(j)}_{\varphi}|\varphi\in\Phi\}\bigr)_{j=1}^J$ be the corresponding family of marginals. 
Suppose that 
\begin{itemize}
\item[(i)] $\Phi$ is a compact metric space.
\end{itemize}
Furthermore, suppose $\varphi_0\in\Phi$ is a distinguishable parameter value in the sense of satisfying
\begin{itemize}
\item[(ii)]  $Q(D_{\varepsilon}(\varphi_0)) > 0$ for every $\varepsilon>0$, and
\item[(iii)] For each $\varphi\neq\varphi_0$ there is an open set $N_{\varphi} \ni \varphi$ such that $\inf_{j\geq 1} H^{(j)}_{\varphi_0}[Z(N_{\varphi},\Pi_j)] > 0$.
\end{itemize}

Then for every $\varepsilon>0$ and open set $C_0 \subset\Phi$ containing $D_{\varepsilon}(\varphi_0)$,
the superposterior satisfies $Q^*(C_0|\Pi_1\ldots,\Pi_J)\rightarrow 1$ almost surely as $J\rightarrow\infty$.
\end{theorem}

\begin{proof}
Let $\Pi = (\Pi_1,\Pi_2,\ldots)\in\mathcal{P}^{\infty}$, where $\mathcal{P}=\mathcal{P}(\Theta)$ is the space of probability measures on $\Theta$.

Recall that $H^{(j)}_{\varphi}$ is the marginal conditional distribution of $\Pi_j$ given $\varphi$, and define the product distribution $H_{\varphi} = \bigotimes_{j=1}^{\infty} H^{(j)}_{\varphi}$ over $\mathcal{P}^{\infty}$.

As defined in Section~\ref{sec:posterior_concentration}, let $h_j(\Pi_j|\varphi)$ be the density of $H^{(j)}_{\varphi}$ with respect to some base measure.
For any $J \geq 1$ define
\[
L_J(\varphi,\Pi)
= \frac{1}{J} \sum_{j=1}^{J} \log\left( \frac{h_j(\Pi_j|\varphi_0)}{h_j(\Pi_j|\varphi)} \right).
\]

Given $\varepsilon>0$, let $C_0 \subset\Phi$ be an open set containing $D_{\varepsilon} = D_{\varepsilon}(\varphi_0)$, and let $C_0^c = \Phi\setminus C_0$ denote the complement of $C_0$ (with respect to the space $\Phi$).

The `posterior odds' of $C_0$ is
\begin{align}
    \frac{Q^*\left(C_0|\Pi_1,\ldots,\Pi_J\right)}%
    {Q^*\left(C_0^c|\Pi_1,\ldots,\Pi_J\right)} 
    &=\frac{\int_{C_0} Q^*\left(d\varphi|\Pi_1,\ldots,\Pi_J\right)}%
    {\int_{C_0^c} Q^*\left(d\varphi|\Pi_1,\ldots,\Pi_J\right)} 
    \label{eq:posterior_odds}
    \\[1mm]
    &= \frac{\int_{C_0} \prod_{j=1}^J h_j(\Pi_j|\varphi) Q(d\varphi)}%
    {\int_{C_0^c} \prod_{j=1}^J h_j(\Pi_j|\varphi) Q(d\varphi)} \nonumber
    \\[1mm]
    &= \frac{\int_{C_0} \exp\{-J L_J(\varphi,\Pi)\} Q(d\varphi)}%
    {\int_{C_0^c} \exp\{-J L_J(\varphi,\Pi)\} Q(d\varphi)}.
    \label{eq:posterior_odds_other}
\end{align}

The idea of the proof is to find a lowed bound on the numerator of Eq.~\eqref{eq:posterior_odds_other} and an upper bound on the denominator of Eq.~\eqref{eq:posterior_odds_other} such that their ratio goes to $\infty$.

For the denominator of Eq.~\eqref{eq:posterior_odds_other}, we construct a finite list of open sets $C_1,\ldots,C_m \subset \Phi$ so that $C_0^c \subset C_1 \cup \cdots \cup C_m$, and let $c_i := \inf_{j\geq 1} H^{(j)}_{\varphi_0}[Z(C_i,\Pi_j)]$ so that $c_i>0$ by assumption. For any $A \subset \Phi$ we have
\[
\inf_{\varphi \in A} L_J(\varphi,\Pi)
\geq \frac{1}{J} \sum_{j=1}^{J} Z(A,\Pi_j).
\]
Then the denominator in Eq.~\eqref{eq:posterior_odds_other} is at most
\begin{align*}
    \sum_{i=1}^{m} \int_{C_i} \exp\{-J L_J(\varphi,\Pi)\} Q(d\varphi)
    \leq \sum_{i=1}^{m} \sup_{\varphi\in C_i} \exp\{-J L_J(\varphi,\Pi)\} Q(C_i)
    \\[1mm]
    \leq \sum_{i=1}^{m} \exp\left\{-\sum_{j=1}^{J} Z(C_i; \Pi_j) \right\} Q(C_i).
\end{align*}
For each $i \in [m]$, by the strong law of large numbers, there exists a set $B_i \subset \mathcal{P}^{\infty}$ such that $H_{\varphi_0}(B_i) = 1$ and such that, for every $\Pi \in B_i$ there exists an integer $K_i(\Pi)$ such that $J \geq K_i(\Pi)$ implies $J^{-1}\sum_{j=1}^{J} Z(C_i;\Pi_j) > c_i/2 >0$.
Let $c = \min\{ c_1,\ldots,c_m \}$, $B = \cap_{i=1}^{m} B_i$, and $K(\Pi) = \max\{ K_1(\Pi),\ldots,K_m(\Pi) \}$.
Then for each $\Pi \in B$ and $J \geq K(\Pi)$, the denominator in Eq.~\eqref{eq:posterior_odds_other} is at most $\exp(-Jc/2)$.

For the numerator of Eq.~\eqref{eq:posterior_odds_other}, let $\delta\in(0,\min\{\varepsilon,c/2\})$. For each $\Pi$ and $\varphi$, define 
\begin{align*}
    W_J(\Pi) &= \{ \varphi \ : \ L_k(\varphi,\Pi) \leq I_k(\varphi_0;\varphi) + \delta \ \text{ for all } k \geq J \},
    \\
    V_J(\varphi) &= \{ \Pi \ : \ L_k(\varphi,\Pi) \leq I_k(\varphi_0;\varphi) + \delta \ \text{ for all } k \geq J \}.
\end{align*}
For convenience we denote $D_{\delta} = D_{\delta}(\varphi_0)$.
For each $\varphi$, by the strong law of large numbers we have $L_J(\varphi,\Pi) - I_J(\varphi_0;\varphi) \to 0$  almost surely [$H_{\varphi_0}$], so then $H_{\varphi_0}(\cup_{J=1}^{\infty} V_J(\varphi)) = 1$. This fact, together with the fact that the sets $V_J(\varphi)$ are increasing ($V_{J}(\varphi) \subset V_{J+1}(\varphi)$) and the fact that $\Pi\in V_J(\varphi)$ if and only if $\varphi \in W_J(\Pi)$, allow us to write
\begin{align*}
    Q(D_{\delta})
    &= \int_{D_{\delta}} H_{\varphi_0} \left( \bigcup_{J=1}^{\infty} V_J(\varphi)\right) Q(d\varphi)
    \\
    &= \lim_{J \to\infty} \int_{D_{\delta}} H_{\varphi_0}(V_J(\varphi)) Q(d\varphi)
    \\
    &= \lim_{J \to\infty} \int_{D_{\delta}} \int_{\mathcal{P}^{\infty}} \mathbb{I}[\Pi\in V_J(\varphi)] H_{\varphi_0}(d\Pi) Q(d\varphi)
    \\
    &= \lim_{J \to\infty} \int_{\mathcal{P}^{\infty}} \int_{D_{\delta}} \mathbb{I}[\Pi\in V_J(\varphi)] Q(d\varphi) H_{\varphi_0}(d\Pi)
    \\
    &= \lim_{J \to\infty} \int_{\mathcal{P}^{\infty}} \int_{D_{\delta}} \mathbb{I}[\varphi\in W_J(\Pi)] Q(d\varphi) H_{\varphi_0}(d\Pi)
    \\
    &= \lim_{J \to\infty} \int_{\mathcal{P}^{\infty}} Q(D_{\delta} \cap W_J(\Pi)) H_{\varphi_0}(d\Pi)
    \\
    &= \int_{\mathcal{P}^{\infty}} \lim_{J \to\infty} Q(D_{\delta} \cap W_J(\Pi)) H_{\varphi_0}(d\Pi).
\end{align*}
Now notice that since $Q(D_{\delta} \cap W_J(\Pi)) \leq Q(D_{\delta})$ for all $\Pi$ and $J$, we then have
\[
\lim_{J \to\infty} Q(D_{\delta} \cap W_J(\Pi))
= Q(D_{\delta}) \ \text{ almost surely } [H_{\varphi_0}]
\]
because strict inequality with positive probability would contradict the above string of equalities (this is a measure-theory argument). So, there is a set $B' \subset\mathcal{P}^{\infty}$ with $H_{\varphi_0}(B') = 1$ and for every $\Pi\in B'$ there exists a positive integer $K'(\Pi)$ such that $J \geq K'(\Pi)$ implies $Q(D_{\delta} \cap W_J(\Pi)) > Q(D_{\delta})/2$.
So, if $\Pi\in B'$ and $J \geq K'(\Pi)$, then the numerator in Eq.~\eqref{eq:posterior_odds_other} is at least
\begin{align*}
    \int_{D_{\delta} \cap W_J(\Pi)} \exp\{-J[I(\varphi_0;\varphi) + \delta]\} Q(d\varphi)
    &\geq \frac{1}{2}\exp(-2J\delta) Q(D_{\delta})
    \\
    &\geq \frac{1}{2}\exp(-Jc/4) Q(D_{\delta}),
\end{align*}
where the first inequality is because $I_J(\varphi_0;\varphi) \leq \delta$ for $\varphi\in D_{\delta}$, and the second one is because $\delta < c/2$
by the choice of $\delta$.

It follows that if $\Pi \in B \cap B'$ and $J \geq \max\{ K(\Pi), K'(\Pi) \}$, then the ratio in Eq.~\eqref{eq:posterior_odds_other} is at least $Q(D_{\delta})\exp(Jc/4)/2$, which goes to $\infty$ with $J$.

To finish the proof, notice that the ratio in Eq.~\eqref{eq:posterior_odds_other} equals the odds ratio in Eq.~\eqref{eq:posterior_odds}, meaning $Q^*\left(C_0|\Pi_1,\ldots,\Pi_J\right)/Q^*\left(C_0^c|\Pi_1,\ldots,\Pi_J\right)$, so then $Q^*\left(C_0^c|\Pi_1,\ldots,\Pi_J\right) \to 0$ almost surely [$H_{\varphi_0}$] when $J\to\infty$, and necessarily $Q^*\left(C_0|\Pi_1,\ldots,\Pi_J\right) \to 1$.
\end{proof}

The interested readers may take note that this proof follows similar steps as the proof of \citet[Theorem~7.80]{Schervish:1995}, with suitable adaptations to our setting.

Some comments regarding the `distinguishable $\varphi_0$' satisfying conditions (ii) and (iii) of our Theorem~\ref{thm:concentration_continuous}: This parameter value receives positive probability mass from the prior (condition (ii)) and it is a unique minimizer of the KL functional (condition (iii)), which are standard requirements for ensuring model identifiability and posterior consistency \citep[e.g.][]{Bernardo+Smith:1994,Schervish:1995,Van-der-Vaart-1998,Gelman+others:2013}.

\section{Interpretation of MBA as message passing}\label{sec:bp}

The central updates of the MBA framework, Equations (\ref{eq:posterior_uncertain_obs_density}) and (\ref{eq:update_pi_j}) in the main text, 
can also be viewed within the formalism of probabilistic graphical models. 
This provides both an intuitive interpretation and a visualization 
of the updates,  
and gives a straightforward way of extending the framework to more complex model structures. 
To elaborate further on this, 
consider a tree-structured undirected graphical model with $J$ leaf nodes and a root. 
This is a special case of a pairwise Markov network \citep{Koller+Friedman:2009}, where all factors are associated with single variables or pairs of variables, referred to as \emph{node potentials} and \emph{edge potentials}, respectively. Note that the potential functions are simply non-negative functions, which may not integrate to 1. The joint distribution is the normalized product of all potentials. 
Choosing the node potentials as $\pi_j(\theta_j), j=1,\ldots,J$ and $q(\varphi)$, and the edge potentials as $p(\theta_j|\varphi)$, the model has the joint density
\begin{equation*}
\frac{1}{Z}q(\varphi)\prod_{j=1}^J p(\theta_j|\varphi)\pi_j(\theta_j),
\end{equation*}
where $Z$ is a normalizing constant. 
Finding the  
marginal density of $\varphi$ can then be interpreted as propagating beliefs from each of the leaf nodes up to the root node in the form of messages, a process known as \emph{message passing} or belief propagation \cite[e.g.][]{Yedidia+others:2001}.
To that end, we specify the following messages from the $j$th leaf node to the root:
\begin{equation}\label{eq:leaf_root_message}
m_{\theta_j\rightarrow \varphi}(\varphi)\propto \int p(\theta_j|\varphi)\pi_j(\theta_j) d \theta_j .
\end{equation}
 The initial belief $q(\varphi)$ on $\varphi$ is then updated according to
\begin{equation}\label{eq:bp_vaprhi_marginal}
q^*(\varphi) \propto q(\varphi)\prod_{j=1}^J m_{\theta_j\rightarrow \varphi}(\varphi),
\end{equation}
which is exactly equal to Equation~\eqref{eq:posterior_uncertain_obs_density}, and illustrated in Figure~\ref{fig:bp_varphi}.

In a similar way, we may pass information to any single leaf node from the remaining leaf nodes. 
We now specify two kinds of messages: from leaf nodes indexed by $j\in\J\setminus j'$ to the root node, as given by Equation~\eqref{eq:leaf_root_message}, and from the root node to the $j'$th leaf node,
\begin{equation*}\label{eq:root_leaf_message}
m_{\varphi\rightarrow \theta_{j'}}(\theta_{j'})
\propto \int_{\Phi}p(\theta_{j'}|\varphi)q(\varphi)\prod_{j\in\J\setminus j'} m_{\theta_j\rightarrow \varphi}(\varphi)\, d\varphi.
\end{equation*}
The updated belief over $\theta_{j'}$ is then 
\begin{equation}\label{eq:bp_theta_marginal}
\pi_{j'}^*(\theta_{j'}) \propto 
\pi_{j'}(\theta_{j'}) m_{\varphi\rightarrow \theta_{j'}}(\theta_{j'}),
\end{equation}
which is exactly equal to Equation~\eqref{eq:update_pi_j}, and illustrated in Figure~\ref{fig:bp_theta}. 

\begin{figure}[t!]
    \centering
    \begin{subfigure}[t]{0.475\linewidth}
        \centering
        \includegraphics[width=\linewidth]{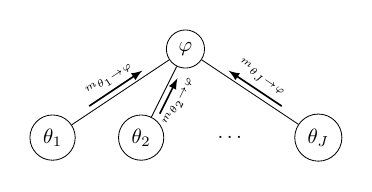}
        \caption{~}
        \label{fig:bp_varphi}
    \end{subfigure}%
    ~ 
    \begin{subfigure}[t]{0.475\linewidth}
        \centering
\includegraphics[width=\linewidth]{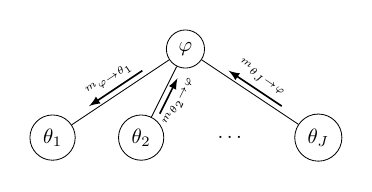}    
  \caption{~}
                \label{fig:bp_theta}
    \end{subfigure}
    \caption{Propagating beliefs by messages-passing (a) from each of the leaf nodes to the root node and (b) from each of the nodes $\theta_2,\ldots,\theta_J$ to $\varphi$, and then finally, from $\varphi$ to $\theta_1$. The updated beliefs over $\varphi$ and $\theta_1$ are $q^*(\varphi) \propto q(\varphi)\prod_{j=1}^J m_{\theta_j\rightarrow \varphi}(\varphi)$ and $\pi_{1}^*(\theta_{1}) \propto \pi_{1}(\theta_{1})\, m_{\varphi\rightarrow \theta_{1}}(\theta_{1})$, respectively.}
\end{figure}

Although not utilized in this work, the graphical model view may also be useful in devising efficient computational strategies. Especially with more complex model structures, making use of the conditional independencies made explicit by the graphical model may bring considerable computational gains.

\section{Computational strategy}\label{sec:computation}

Here we describe a simple computational strategy, which is used in the numerical examples of Section~\ref{sec:numerical_demos} in the main article. 
Some further alternatives are briefly discussed at the end of this section. 
Recall now that the density of the joint distribution of the parameters $\theta_1,\ldots,\theta_J,\varphi$ can be written as 
\begin{equation}\label{eq:joint_density}
\frac{1}{Z}q(\varphi)\prod_{j=1}^J p(\theta_j|\varphi)\pi_j(\theta_j).
\end{equation}
Our goal is to produce joint samples from the above model, enabling any desired marginals to be extracted from them. 
Probabilistic programming languages \citep[e.g.][]{Carpenter+others:2017, Salvatier+others:2015, Tran+others:2016, Wood+others:2014} allow  sampling from an arbitrary model, provided that the components of the (unnormalized) model can be specified in terms of probability distributions of some standard form. 
In the illustrations of this section, we use Hamiltonian Monte Carlo implemented in the Stan probabilistic programming language \citep{Carpenter+others:2017}.

We first note that in the above joint model (\ref{eq:joint_density}), the part specified by the meta-analyst, i.e. $q(\varphi)\prod_{j=1}^J p(\theta_j|\varphi)$, can by design be composed using standard parametric distributions.  
The observed part of the model $\prod_{j=1}^J \pi_j(\theta_j)$, however, is in general analytically intractable, and instead of having direct access to posterior density functions of standard parametric form, we typically have a sets of posterior samples $\left\{\theta_j^{(1)},\ldots,\theta_j^{(L_j)}\right\}$, with $\theta_j^{(l)}\sim \Pi_j$.
Our strategy is then to first find an intermediate parametric approximation $\hat{\pi}_j$ for $\pi_j$, which enables us to sample from an \emph{approximate} joint distribution. Assuming that the true densities $\pi_j(\theta_j)$ can be evaluated using e.g. kernel density estimation, and that $\hat{\pi}_j(\theta_j)=0 \Rightarrow \pi_j(\theta_j)=0$, the joint samples can be further \emph{refined} using sampling/importance resampling \citep[SIR;][]{Smith+Gelfand:1992}. 
The steps of the computational scheme are summarized below:
\begin{enumerate}
  \item For $j=1,\ldots,J$, fit a parametric density function $\hat{\pi}_j$ to the samples 
  $\left\{\theta_j^{(1)},\ldots,\theta_j^{(L_j)}\right\}$.
  \item Draw $M$ samples $\mathcal{S} = \left\{\theta_1^{*\,(m)},\ldots,\theta_J^{*\,(m)}, \varphi^{*\,(m)}\right\}_{m=1}^{M}$ from the approximate joint model $\frac{1}{Z'}q(\varphi)\prod_{j=1}^J p(\theta_{j}|\varphi)\hat{\pi}_j(\theta_j)$.
  \item Compute importance weights $w_m=\tilde{w}_m/\sum_{m=1}^M\tilde{w}_m$, where
\begin{align*}
\tilde{w}_m 
&= \frac{Z^{-1}\, q\left(\varphi^{*\,(m)}\right)\prod_{j=1}^J p\left(\theta_j^{*\,(m)}|\varphi^{*\,(m)}\right)\pi_j\left(\theta_j^{*\,(m)}\right)}
{(Z')^{-1}\, q\left(\varphi^{*\,(m)}\right)\prod_{j=1}^J p\left(\theta_j^{*\,(m)}|\varphi^{*\,(m)}\right)\hat{\pi}_j\left(\theta_j^{*\,(m)}\right)}
= \frac{Z'\,\prod_{j=1}^J\pi_j\left(\theta_j^{*\,(m)}\right)}{Z\,\prod_{j=1}^J\hat{\pi}_j\left(\theta_j^{*\,(m)}\right)}. 
\end{align*}
Note that the constant $Z'/Z$ cancels in the computation of the normalized weights $w_m$. 
 \item Resample $\mathcal{S}$ with weights $\{w_1,\ldots,w_M\}$. 
\end{enumerate}

For problems with a very large number of studies or high dimensional local parameters, or if the imposed parametric densities approximate the actual posteriors poorly, the computation of importance weights may become numerically unstable. 
The issue could possibly be mitigated using more advanced importance sampling schemes, such as Pareto-smoothed importance sampling \citep{Vehtari+others:2015} or prior swap importance sampling \citep{Neiswanger+Xing:2017}.
If we are only interested in sampling from the density of the global parameter, as given by Equation~\eqref{eq:posterior_uncertain_obs_density} in the main text, then an obvious alternative strategy would be to implement a Metropolis-Hastings algorithm, using the samples $\theta_j^{(l)}\sim \Pi_j$ to compute Monte Carlo estimates of the integrals $\int_{\Theta} p(\theta_j|\varphi)\pi_j(\theta_j) d\theta_j$. However, this would lead to expensive MCMC updates as the integrals need to be re-estimated at every iteration of the algorithm. 

Finally, 
instead of directly sampling from the full joint distribution, we could try to utilize the induced graphical model structure (Appendix~\ref{sec:bp} above) to do localized inference. 
Unlike our current strategy, which infers all parameters of the joint model simultaneously, this would enable us to devise more efficient and flexible computational schemes.
In particular, developments in nonparametric and particle-based belief propagation \citep{Ihler+McAllester:2009, Lienart+others:2015, Pacheco+others:2014, Sudderth+others:2010} represent a promising direction in this line of work. Alternatively, multi-stage sampling schemes, such as the one proposed by \citet{Goudie2019}, could be explored.

\section{Synthetic tuberculosis example}
\label{sec:tbc:tbc_synthetic}

We also provide a simulated version of the tuberculosis experiments in Section~\ref{subsec:tbc} of the main paper.
We set the true value of the overall mean effect $\mu = (\mu_1, \mu_2)$ to $(2.3, 0.95)$. Then, we sample the local effects $R_1$ from $\mathcal{N}(\mu_1, 0.25)$ and $R_2$ from  $\mathcal{N}(\mu_2, 0.05)$. 

Similarly to Section~\ref{subsec:tbc}, we first define the combination model as if the local effects $\theta_j=(R^{(1)}_{0,j}, R^{(2)}_{0,j})$, for $j=1\dots J$, are observed quantities:
\begin{gather*}
    \theta_j\sim \mathcal{N}_2(\mu, \Sigma_0), \quad
    \mu_1 \sim \mathrm{Gamma}(a_1, b_1),\quad
    \mu_2 \sim \mathrm{Gamma}(a_2, b_2), \quad
    \Sigma_0\sim \mathcal{W}^{-1}(\nu,\Psi). 
\end{gather*}
We set the hyper-parameters as follows:
\[
a_1 = 2.0,\, b_1 = 1.0\quad 
a_2 = 2.0, \, b_2 = 2.0,\quad  \text{ and } \quad
\nu = 4,\quad  
\Psi = \begin{bmatrix}
0.2 & 0 \\
0 & 0.01
\end{bmatrix}.
\]
Then, we use the MBA update rules to incorporate the study-specific posteriors, computed independently using ABC.
Figure~\ref{fig:local_tb_synth} shows the updated beliefs for the local effects $\theta_1, \ldots, \theta_J$, after borrowing strength across the individual studies. Similarly to what was observed with real world-data (Section~\ref{subsec:tbc} in the main paper), the updated posteriors exhibit a more identifiable behavior.  
Figure~\ref{fig:global_tb_synth} shows that the superposterior for the overall mean effect concentrates close to the true value $\mu$.

\begin{figure}[H]
\centering
\includegraphics[width=0.72\linewidth]{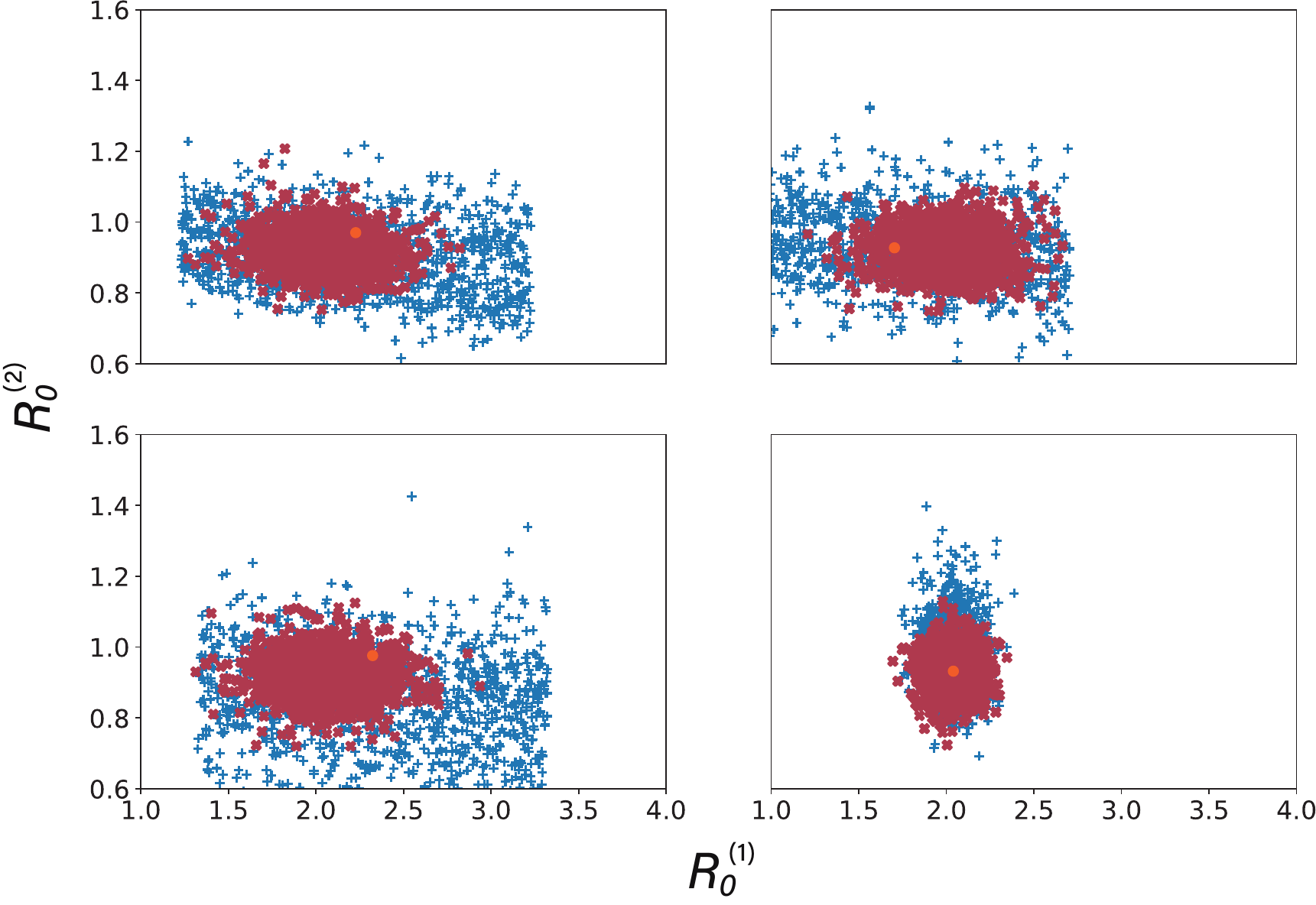}\\
\caption{Synthetic tuberculosis: posteriors for the reproductive numbers $(R^{(1)}_{0,j}, R^{(2)}_{0,j})$ updated using MBA (red), plotted on top of the original, individually obtained posteriors (blue). The orange dots denote the true value of local effects.}\label{fig:local_tb_synth}
\end{figure}

\begin{figure}[H]
\centering
\includegraphics[width=0.45\linewidth]{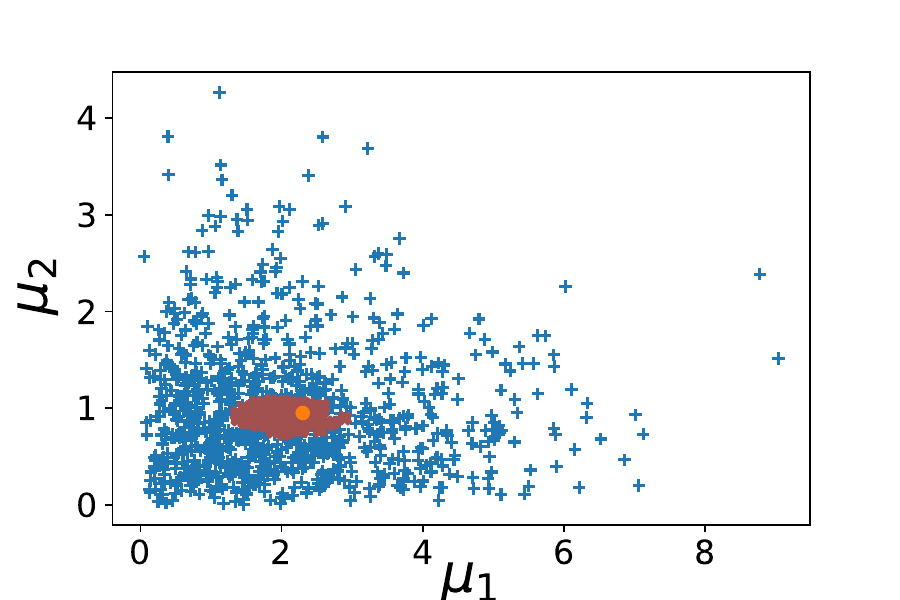}\\
\caption{Synthetic tuberculosis: superposterior (red) for the overall effect $\mu=(\mu_1, \mu_2)$, plotted on top of its prior (blue). The orange dot denotes the true value of overall effect.}\label{fig:global_tb_synth}
\end{figure}

\paragraph{Details on the ABC simulations.} To simulate the same scenario as in the real-world experiment, we use an informative prior to compute the ABC posterior for one of the local studies (bottom right in Figure~\ref{fig:local_tb_synth}). More specifically, we sample $R^{(1)}_0$ from a Normal with standard deviation $0.1$ and sample $R^{(2)}_0$ from an uniform with variance $1/12$, each centered around the true value of $R^{(1)}_0$ and $R^{(2)}_0$, respectively. For the other local ABC studies, we use a uniform prior with support on a $2\times 1$ rectangle centered on $\theta_j$. For all simulations, we set the burden rate to $10$ and the net transmission rate of the non-compliant population to $11$. The remaining implementation details are identical to the ones for the real-world case.

\end{document}